\documentclass[10pt]{article}
\usepackage[papersize={8.5in,11in},top=0.5in,left=1in,right=1in,bottom=1in]{geometry}
\usepackage{authblk}
\usepackage[square,numbers]{natbib}

\usepackage{amsmath}
\usepackage{hyperref}
\usepackage{amssymb}
\usepackage{cleveref}
\usepackage{amsthm}
\newcommand\longvdots[1]{\raisebox{1em}{\rotatebox{-90}{\hbox to #1 {\dotfill}}}}
\usepackage{algorithm}
\usepackage{algpseudocode}
\usepackage{tikz}
\usepackage{amsmath}
\usepackage{xcolor}

\newtheorem{assumption}{Assumption}

\newtheorem{theorem}{Theorem}

\newtheorem{lemma}{Lemma}
\newtheorem{remark}{Remark}

\DeclareMathOperator{\E}{\mathbb{E}}

\title{Multi Timescale Stochastic Approximation: \\Stability and Convergence}
\author[1]{Rohan Deb}
\author[2,3]{Swetha Ganesh}
\author[2]{Shalabh Bhatnagar}

\affil[1]{University of Illinois, Urbana-Champaign}
\affil[2]{Indian Institute of Science, Bengaluru}
\affil[3]{Purdue University, West Lafayette}

\begin{document}

\maketitle
\begin{abstract}

This paper presents the first sufficient conditions that guarantee the stability and almost sure convergence of multi-timescale stochastic approximation (SA) iterates. It extends the existing results on one-timescale and two-timescale SA iterates to general $N$-timescale stochastic recursions, for any $N\geq1$, using the ordinary differential equation (ODE) method. 
{\color{black} As an application of our results, we first study stochastic approximation algorithms augmented with heavy-ball momentum in the context of Gradient Temporal Difference (GTD) learning. The addition of momentum introduces an auxiliary state that evolves on an intermediate timescale, resulting in a three-timescale recursion. We show that when the momentum parameters are chosen appropriately, the resulting scheme fits within our framework and converges almost surely to the same fixed point as the baseline GTD algorithm. The stability and convergence of all iterates—including the momentum state—are guaranteed by our main results, without requiring ad hoc bounds.
We then study off-policy actor–critic algorithms with a baseline learner, actor, and critic updated on separate timescales. In contrast to prior work, we eliminate projection steps from the actor update and instead use our convergence framework to guarantee stability and almost sure convergence of all components. 
Finally, we extend our analysis to constrained policy optimization in the average-reward setting, where the actor, critic, and dual variables evolve on three distinct timescales, and we verify that the resulting dynamics satisfy the conditions of our general theorem. Together, these examples demonstrate how diverse reinforcement learning algorithms—spanning momentum acceleration, off-policy learning, and primal–dual methods—fit naturally into the proposed multi-timescale framework.}
\end{abstract}

\section{Introduction}
\label{intro}
Stochastic Approximation (SA) algorithms \citep{robbinsmonro} are a class of iterative procedures used to find the zeros of a function $h:\mathbb{R}^d\rightarrow \mathbb{R}^d$, when the actual function is not known but noisy observations of the same are available. Standard stochastic approximation iterates are as given below:
\begin{equation}
\label{StandardSA}
    x_{n+1} = x_{n} + \alpha_n f(x_{n},Y_{n}), n\geq 0,
\end{equation}
where, $x_{n}\in\mathbb{R}^{d}$, $n\geq 0$, is a sequence of parameters updated as in (\ref{StandardSA}) using $f(x_{n},Y_{n})$ that are noisy observations of the true function value $h(x_n)$. Further, $Y_{n}$, $n\geq 0$, is the corresponding noise sequence and $\alpha_n, n\geq 0$, the given step-size sequence. In this paper, we assume that the noisy observation of the function is of the following form:
\begin{equation}
\label{eqf}
    f(x_{n},Y_{n}) = h(x_{n}) + M_{n+1},
\end{equation}
where $h(\cdot)$ is the objective function and $M_{n+1}, n\geq 0$, is a martingale difference sequence with respect to the increasing sequence of sigma fields $\mathcal{F}_{n} = \sigma(x_{i},M_{i}, 0\leq i \leq n$), $n\geq 0$. The form (\ref{eqf}) of $f$ is fairly general. For instance, if
$Y_n \neq 0$, are i.i.d., one may let $h(x_n) = E_Y[f(x_n,Y)]$ and $M_{n+1} = f(x_n,Y_n) - h(x_n),$ $n\geq 0$, are the corresponding martingale differences.
Then (\ref{eqf}) can be seen to hold. 

In many applications, the objective function $h$ might require an additional averaging for it's evaluation at a given parameter update. The procedure, in such a case, would require an additional recursion to estimate $h$ and the overall procedure would result in a recursive computation in nested loops. For example, the policy iteration procedure in Markov decision processes \citep{puterman2014markov},
involves two nested loops in which the outer loop performs policy improvement while the inner loop performs policy evaluation for a given policy update. 
Actor-critic algorithms \citep{konda2000actor, bhatnagar2009natural} are stochastic recursive iterates that try to mimic the nested loop structure of policy iteration under noisy sample estimates where the actor recursion performs policy improvement while the critic is responsible for policy evaluation. This is achieved 
by making use of two different step-size schedules, one of which goes to zero at a rate faster than the other. The recursion governed by the step-size that goes to zero faster ends up being the `slower' timescale recursion while the one governed by the other step-size schedule corresponds to the `faster' timescale update. Thus, even though both recursions are executed simultaneously, at each instant, one sees a similar effect that a nested loop procedure would provide. In particular, the faster recursion sees the slower update as `quasi-static' while the latter sees the faster recursion as essentially `equilibrated'.  
  
Another example along these lines would be of a Markov reward process \citep{marbach2001simulation} whose transition probabilities depend on a certain parameter $\theta$ and the goal is to find the optimal parameter, i.e., the one that maximizes the long-run average reward. Two-timescale stochastic iterates can be used to solve such problems as well \citep{s2013stochastic} where the recursion corresponding to the inner loop, i.e., the one governed by the faster timescale, performs the averaging of the single-stage rewards for a given update of the parameter $\theta$ while the outer loop recursion or the one
governed by the slower timescale performs a maximization of the average rewards over $\theta$. 
Two-timescale stochastic approximation algorithms have the following general form:
\begin{gather}
    \label{TTS_1}
    x^{(1)}_{n+1} = x^{(1)}_{n} + \alpha^{(1)}_{n}\left(h^{(1)}\left(x^{(1)}_{n},x^{(2)}_{n}\right) + M^{(1)}_{n+1}\right),\\
    \label{TTS_2}
    x^{(2)}_{n+1} = x^{(2)}_{n} + \alpha^{(2)}_{n}\left(h^{(2)}\left(x^{(1)}_{n},x^{(2)}_{n}\right) + M^{(2)}_{n+1}\right),
\end{gather}
where $\{M^{(1)}_{n+1}\}$ and $\{M^{(2)}_{n+1}\}$ are both martingale difference sequences and $\{\alpha^{(1)}_{n}\}$ and $\{\alpha^{(2)}_{n}\}$ are the step-size sequences, respectively. A specific requirement of the step-sizes for (\ref{TTS_1})-(\ref{TTS_2}) to be a two-timescale recursion is that $\frac{\alpha^{(1)}_{n}}{\alpha^{(2)}_{n}} \rightarrow 0$ as $n\rightarrow\infty$.

These algorithms have been analyzed for their convergence in \cite{borkar_tts} assuming, in particular, that the iterates remain stable, i.e., that almost surely $\sup_n \parallel x^{(1)}_n\parallel <\infty$ and $\sup_n \parallel x^{(2)}_n\parallel <\infty$, respectively. This is a nontrivial requirement and in fact until recently, there were no known conditions to verify this requirement. In the case of one-timescale algorithms, i.e., when $N=1$, sufficient conditions for stability and convergence of stochastic iterates are available in the literature. For instance, 
in \citep{Borkar_Meyn,BorkarBook}, certain conditions are provided on the limiting ODE and its scaling limit that are shown to be sufficient for the stability of the recursion. 

In \cite{chandru-SB}, the stability conditions provided in \cite{Borkar_Meyn} are generalized to the two-timescale setting and these happen to be the first set of conditions that guarantee stability of general two-timescale stochastic recursions and along with \cite{borkar_tts}, provide a complete set of conditions under which the recursions given by \eqref{TTS_1} and \eqref{TTS_2} converge almost surely as $n\rightarrow\infty$. In this paper, we generalize the results of \cite{chandru-SB} to the case of $N$-timescale stochastic approximation for any general $N\geq 1$. We explain below the motivation for doing so.

In many problems of optimization or control under uncertainty, one encounters algorithms involving three or more timescale recursions. Such examples arise for instance when designing actor-critic algorithms for constrained Markov decision processes \citep{altman1999constrained}, see for instance, \cite{borkar2005actor, bhatnagar2010actor} or hierarchical reinforcement learning \citep{bhatnagar-panigrahi}.
In \cite{borkar2005actor, bhatnagar2010actor}, the constraints are relaxed using a Lagrangian and the Lagrange parameter is updated on the slowest timescale, while the actor (policy improvement) update
happens on the middle timescale and the critic (policy evaluation) update happens on the fastest timescale. 
Further, for settings such as hierarchical reinforcement learning \citep{bhatnagar-panigrahi}, one may encounter algorithms with a number of timescales that is proportionate to the number of levels of decision making. 
Sufficient conditions for stability and convergence of algorithms (with $N>2$) are not available in the literature. 
The results we present here are general enough and applicable to such settings. In fact, we also study, in this paper, applications involving certain reinforcement learning algorithms with momentum that require three and four timescale recursions whose convergence we prove by verifying that our sufficient conditions hold.

We provide the first set of sufficient conditions that ensure both the stability and convergence of general $N$-timescale stochastic approximation recursions. Our conditions are obtained by generalizing the requirements in \cite{Borkar_Meyn} for one-timescale algorithms as well as in \cite{chandru-SB} for two-timescale recursions. We provide the full analysis of stability and convergence of a general $N$-timescale stochastic recursion using the aforementioned sufficient conditions.
We then demonstrate the usefulness of these conditions in the context of three and four timescale gradient temporal difference learning algorithms in reinforcement learning. Such an analysis is made possible because of the sufficient conditions that we provide for the stability and convergence of $N$-timescale algorithms for any $N\geq 1$.
Since $N\geq 1$ is arbitrary, we believe that the sufficient conditions for stability and convergence of $N$-timescale stochastic recursions will be extremely useful for algorithm designers who will henceforth be able to prove the stability and convergence of the stochastic recursions through verification of the aforementioned sufficient  conditions.

The rest of the paper is organized as follows: In Section~\ref{notassum}, we present the notation used in the paper as well as our assumptions. Here, for improved clarity, we first present our assumptions for the 3-timescale case before we present the general assumptions for $N$-timescale stochastic iterates. In Section~\ref{mainresults}, we present our main results on stability and convergence of $N$-timescale SA. {\color{black} We then provide brief descriptions of several applications of our results in reinforcement learning and provide a proof sketch of the main results.} Next in  Section~\ref{proof}, we provide the complete analysis of stability and convergence of these recursions. {\color{black} In Section~\ref{sec_app}, we present complete stability and convergence analysis of several reinforcement learning algorithms using our new results for multi-timescale updates. We first start with a policy evaluation algorithm, Gradient Temporal Difference (GTD) learning with heavy-ball momentum, where the momentum term introduces an intermediate timescale and yields a three-timescale recursion. When the momentum parameters are properly chosen, our framework guarantees almost sure convergence of all iterates—including the momentum state—to the standard GTD solution without requiring additional boundedness assumptions. We also show empirically that the proposed momentum methods outperform their vanilla counterparts. 
Next we apply our framework to off-policy actor–critic algorithms with a baseline, where the actor, critic, and baseline updates evolve on separate timescales. Without using projections, we show that the iterates remain stable and converge under standard assumptions by verifying that the recursion satisfies our structural conditions. Finally, we extend our analysis to constrained policy optimization in the average-reward setting, where the actor, critic, and dual variables evolve on three distinct timescales, and we verify that the resulting dynamics satisfy the conditions of our general theorem.} We present concluding remarks as well as potential directions for future research in Section~\ref{conclusion}.

\section{Notations and Assumptions}
\label{notassum}
In this section, we carefully setup the notation followed in the entire paper and then explicitly provide the assumptions under which the iterates are stable and convergent.
The $N$ stochastic recursions are given by:
\begin{gather}
    \begin{aligned}
    x^{(1)}_{n+1} &= x^{(1)}_{n} + \alpha^{(1)}_{n}(h^{(1)}(x^{(1)}_{n},x^{(2)}_{n},\ldots, x^{(N)}_{n}) + M^{(1)}_{n+1}), \\
    x^{(2)}_{n+1} &= x^{(2)}_{n} + \alpha^{(2)}_{n}(h^{(2)}(x^{(1)}_{n},x^{(2)}_{n},\ldots,x^{(N)}_{n}) + M^{(2)}_{n+1}),\\
    &\vdots \\
    x^{(N)}_{n+1} &= x^{(N)}_{n} + \alpha^{(N)}_{n}(h^{(N)}(x^{(1)}_{n},x^{(2)}_{n},\ldots,x^{(N)}_{n}) + M^{(N)}_{n+1}).
    \end{aligned}
\end{gather}
Here, $\{x^{(j)}_{n} \in \mathbb{R}^{d_j}\}_{1\leq j\leq N}$ are the $N$ parameters that are updated at each time-step and the subscript $n$ denotes the time-step or iteration-index  of the update. $\{\alpha^{(j)}_{n}, n\geq 0\}$, $1\leq j \leq N$ are $N$ different step-size sequences. The functions $h^{(j)}:\mathbb{R}^{d_1+d_2+\cdots+d_N}\rightarrow \mathbb{R}^{d_j}$, ${1\leq j\leq N}$ could potentially depend on all the $N$ parameters. Also, $\{M^{(j)}_{n+1}\}$, ${1\leq j \leq N}$ are the associated martingale difference noise sequences.

The following are standard assumptions made while analyzing stochastic recursions:

\begin{itemize}
    \item[\textbf{(A:1)}] $h^{(j)}$, ${1\leq j \leq N}$ are Lipschitz continuous functions.
    
    \item[\textbf{(A:2)}] $\{M^{(j)}_{n+1}\}$, ${1\leq j \leq N}$ are $N$ martingale difference sequences with respect to the increasing sequence of sigma fields $\{\mathcal{F}_{n}\}$, where $\mathcal{F}_{n} = \sigma\left(x^{(j)}_{i},M^{(j)}_{i+1},1\leq j\leq N, i\leq n\right)$, $n\geq 0$, with $\mathbb{E}[\|M^{(i)}_{n+1}\|^{2} | \mathcal{F}_{n}] \leq K^{(i)}(1 + \sum_{j=1}^{N}\|x^{(j)}_{n}\|^{2}), n \geq 0$, for some constant $K^{(i)}>0, 1 \leq i \leq N$.
    \item[\textbf{(A:3)}] $\{\alpha^{(j)}_{n}\}_{1\leq j \leq N}$ are $N$ step-size schedules that satisfy the following:
    \begin{itemize}
        \item[(i)] $\alpha^{(j)}_{n}>0$, $\forall n$,  $\sum_{n}\alpha^{(j)}_{n} = \infty, 1\leq j \leq N$,
        
        \item[(ii)] $\sum_{n=0}^{\infty}\sum_{j=1}^{N}(\alpha^{(j)}_{n})^{2} < \infty$,
        \item[(iii)]
        $\frac{\alpha^{(j)}_{n}}{\alpha^{(j-1)}_{n}} \rightarrow 0,$ as $n\rightarrow\infty, 1 < j \leq N$.
    \end{itemize}
\end{itemize}
Assumptions {\bf (A:1)}-{\bf (A:3)} along with stability of the iterates, i.e., $\sup_{n}(x^{(j)}_{n})<\infty$, $\forall 1\leq j \leq N$ can be shown to ensure convergence. The next set of assumptions along with $\textbf{(A:1)-(A:3)}$ will provide sufficient conditions for stability of the iterates. For the case of one-timescale stochastic update recursions, \cite{Borkar_Meyn} provided verifiable sufficient conditions for stability of the iterate sequence (see (A1) in \cite{Borkar_Meyn} or (A5) in Chapter 3 of \cite{BorkarBook}). 
By considering a scaled version of the original recursion and a scaled ODE associated with the scaled version, \cite{Borkar_Meyn} shows that under sufficiently general conditions, the original recursion remains stable.
 
For two-timescale recursions, \cite{chandru-SB} provided conditions for stability (see (A4) and (A5) there). Since the iterates, in this case, evolve along both the timescales simultaneously, \cite{chandru-SB} analyzed the rescaled trajectories of both the iterates and presented the aforementioned conditions (cf.~(A4) and (A5) therein), one for each timescale, to ensure stability. We extend these ideas to the general $N$-timescale regime, where we analyze the $N$ rescaled iterates and come up with a set of $N$ conditions on the behaviour of the limiting ODEs corresponding to the rescaled recursions.
For improved clarity we first state these conditions in the 3-Timescale setting in Sections~\ref{3TS}. In Section~\ref{sec_app}, we analyze the stability and convergence of a reinforcement learning algorithms with momentum that involves three  timescale updates by verifying these conditions in Sections~\ref{3TS}. We state the conditions for general $N$-timescale stochastic recursions in Section~\ref{NTS}. The assumptions in what follows will be numbered as {\bf (B.j.i)} and {\bf (C.j.i)} respectively. Here the {\bf B}-assumptions concern the rescaled ODEs while the {\bf C}-assumptions are for the ODEs corresponding to the original recursions. Further, the index {\bf j} refers to the number of timescales used in the algorithm and {\bf i} takes values in general between {\bf 1} and {\bf j}. 

\subsection{The 3-Timescale Case}
\label{3TS}

We state first the conditions that ensure stability and convergence for the $N = 3$ case. In Section~\ref{sec_app} we discuss an application of our results to the case of stochastic approximation with momentum where the iterates can be analyzed using the results for the 3-Timescale case.

\begin{enumerate}
    \item [\textbf{(B.3.1)}] For $c\geq 1$, define the following scaled functions based on $h^{(1)}$:
    \begin{gather*}
        h_{c}^{(1)}(x^{(1)},x^{(2)},x^{(3)}) = \frac{h^{(1)}(cx^{(1)},cx^{(2)},cx^{(3)})}{c}.
    \end{gather*}
    Further, $h_{c}^{(1)} \rightarrow h_{\infty}^{(1)}$ as $c\rightarrow \infty$ uniformly on compacts. The ODE
    \(\dot{x}^{(1)}(t) = h^{(1)}_{\infty}(x^{(1)}(t),x^{(2)},x^{(3)}), \)
    for $x^{(j)} \in \mathbb{R}^{d_{j}}$, $j = 2,3$, has a unique globally asymptotically stable equilibrium $\lambda_{\infty}^{(1)}(x^{(2)},x^{(3)}),$ where $\lambda_{\infty}^{(1)}:\mathbb{R}^{d_{2}+d_{3}}\rightarrow\mathbb{R}^{d_1}$ is Lipschitz continuous. Further $\lambda^{(1)}_{\infty}(0,0) = 0$.
    
    \item [\textbf{(B.3.2)}] For $c\geq 1$, define the following scaled functions based on $h^{(2)}$: 
    \begin{gather*}
        h_{c}^{(2)}(x^{(2)},x^{(3)}) = \frac{h^{(2)}\big(c\lambda_{\infty}^{(1)}(x^{(2)},x^{(3)}),cx^{(2)},cx^{(3)}\big)}{c}.
    \end{gather*}
    Further, $h_{c}^{(2)} \rightarrow h_{\infty}^{(2)}$ as $c\rightarrow \infty$ uniformly on compacts. The ODE
    \(\dot{x}^{(2)}(t) = h^{(2)}_{\infty}(x^{(2)}(t),x^{(3)}), \)
    for $x^{(j)} \in \mathbb{R}^{d_{j}}$, $j=3,4$, has a unique globally asymptotically stable equilibrium $\lambda_{\infty}^{(2)}(x^{(3)}),$ where $\lambda_{\infty}^{(2)}:\mathbb{R}^{d_{3}}\rightarrow\mathbb{R}^{d_2}$ is Lipschitz continuous. Further $\lambda^{(2)}_{\infty}(0) = 0$.
    
    \item [\textbf{(B.3.3)}] For $c\geq 1$, define the following scaled functions based on $h^{(3)}$:
    \begin{gather*}
        h_{c}^{(3)}\big(x^{(3)}\big) = \frac{h^{(3)}\left(c\lambda_{\infty}^{(1)}\big(\lambda_{\infty}^{(2)}(x^{(3)}),
        x^{(3)}\big),
        c\lambda_{\infty}^{(2)}(x^{(3)}),cx^{(3)}\right)}{c}.
    \end{gather*}
Further, $h_{c}^{(3)} \rightarrow h_{\infty}^{(3)}$ as $c\rightarrow \infty$ uniformly on compacts. The ODE
    \(\dot{x}^{(3)}(t) = h^{(3)}_{\infty}(x^{(3)}(t))\)
    has the origin in $\mathbb{R}^{d_{3}}$ as its unique globally asymptotically stable equilibrium.
\end{enumerate}

\begin{enumerate}
    \item [\textbf{(C.3.1)}] The ODE
    \begin{equation*}
        \dot{x}^{(1)}(t) = h^{(1)}\big(x^{(1)}(t),x^{(2)},x^{(3)}\big),
    \end{equation*}
    $x^{(j)} \in \mathbb{R}^{d_j}$, $j =2, 3$, has a unique globally asymptotically stable equilibrium $\lambda^{(1)}(x^{(2)},x^{(3)})$ where $\lambda^{(1)}:\mathbb{R}^{d_{2}+d_{3}}\rightarrow \mathbb{R}^{d_{1}}$ is Lipschitz continuous.
    \item [\textbf{(C.3.2)}] The ODE
    \begin{equation*}
        \dot{x}^{(2)}(t) = h^{(2)}\Big(\lambda^{(1)}\big(x^{(2)}(t),x^{(3)}\big),x^{(2)}(t),x^{(3)}\Big),
    \end{equation*}
    with $x^{(3)} \in \mathbb{R}^{d_3}$, has a unique globally asymptotically stable equilibrium $\lambda^{(2)}(x^{(3)})$ and $\lambda^{(2)}:\mathbb{R}^{d_{3}}\rightarrow \mathbb{R}^{d_{2}}$ is Lipschitz continuous.
    \item [\textbf{(C.3.3)}] The ODE
    \begin{equation*}
            \dot{x}^{(3)}(t) = h^{(3)}\Big(\lambda^{(1)}\big(\lambda^{(2)}(x^{(3)}(t)),x^{(3)}(t)\big),
            \lambda^{(2)}(x^{(3)}(t)),x^{(3)}(t)\Big),
    \end{equation*}
    has a unique globally asymptotically stable equilibrium $x^{(3)}_{*} \in \mathbb{R}^{d_3}$.
\end{enumerate}

\subsection{The N-Timescale Case}
\label{NTS}

Next we generalize these assumptions to the $N$-timescale regime. 
The assumptions on the timescales $1,2,\ldots,N$ are compactly stated as assumptions \textbf{(B.N.i)} and \textbf{(C.N.i)}, respectively, where in fact the index $i$ takes $N$ values, i.e., $i=1,\ldots,N$. Thus, these conditions encapsulate a total of $2N$ assumptions, with $N$ B-assumptions and $N$ C-assumptions, respectively. In addition, we shall also require Assumptions A:1-A:3 stated previously.

\begin{enumerate}
    \item [\textbf{(B.N.i)}] For $c\geq 1$, define the following scaled functions based on $h^{(i)}$:
    \begin{gather*}
    h^{(i)}_{c}\left(x^{(i)},x^{(i+1)},\ldots,x^{(N)}\right) = \frac{h^{(i)}(cy^{(1)},cy^{(2)},\ldots,cy^{(N)})}{c},\\
    \text{ where, for $j<i$, } \qquad
        y^{(j)} = \lambda^{j}_{\infty}\Bigg(\lambda^{j+1}_{\infty}\bigg(\ldots\lambda^{(i-3)}_{\infty}\Big(\lambda^{(i-2)}_{\infty}\big(\lambda^{(i-1)}_\infty(x^{(i)},\ldots,x^{(N)}), x^{(i)},\ldots,x^{(N)}\big),\\\lambda^{(i-1)}_\infty(x^{(i)},\ldots,x^{(N)}),x^{(i)},\ldots,x^{(N)}\Big),\ldots,x^{(i)},\ldots,x^{(N)}\bigg)\Bigg)
    \end{gather*}
    and for $j\geq i$
    \(y^{(j)} = cx^{(j)}.\)
    Further, $h^{(i)}_{c}\rightarrow h^{(i)}_{\infty}$ as $c\rightarrow\infty$ uniformly on compacts. The ODE 
    \begin{equation}
        \label{odeB.N.i}
    \dot{x}^{(i)}(t) = h^{(i)}_{\infty}(x^{(i)}(t),x^{(i+1)},\ldots,x^{(N)}),
    \end{equation}
    for $x^{(j)} \in \mathbb{R}^{d_{j}},j\in\{i+1,\ldots,N\}$ has
    \begin{enumerate}
        \item [(i)] a unique globally asymptotically stable equilibrium $\lambda^{(i)}_{\infty}(x^{(i+1)},x^{(i+2)},\ldots,x^{(N)})$, $i=1,\ldots,N-1$, where each $\lambda^{(i)}_{\infty}:\mathbb{R}^{d_{i+1}+d_{i+2}+\cdots+d_N}\rightarrow\mathbb{R}^{d_i}$ is Lipschitz continuous. Further,  $\lambda^{(i)}_{\infty}(0,0,\ldots,0) = 0$, $\forall i=1,\ldots,N-1$, and
        \item [(ii)] for $i=N$, the origin in $\mathbb{R}^{d_{N}}$ is the unique globally asymptotically stable equilibrium of (\ref{odeB.N.i}).
    \end{enumerate}
\end{enumerate}
Chapter 6 of \cite{BorkarBook} makes two more assumptions namely (A1) and (A2) there regarding the global asymptotic stable equilibria of the ODE trajectories along the two timescales. We make $N$ such assumptions here (to account for the $N$ different timescales in our recursions) on the ODE trajectories and compactly write these as \textbf{(C.N.i)},
${1\leq i\leq N}$. 
\vspace{1em}
\begin{enumerate}
    \item [\textbf{(C.N.i)}] The ODE $\dot{x}^{(i)}(t) = h^{(i)}\big(z^{(1)}(t),z^{(2)}(t),\ldots,z^{(N)}(t)\big),$
    has a globally asymptotically stable equilibrium: (i)$\lambda(x^{i+1},\ldots,x^{N})$, $x^{(j)} \in \mathbb{R}^{d_j}, i+1\leq j \leq N$ with $1\leq i \leq N-1$, and (ii) $x^{(N)}_{*} \in \mathbb{R}^{d_N}$.
    \end{enumerate}
    Here, for $j\geq i$, \(z^{(j)}(t) = x^{(j)}.\), and for $j<i,$
    \vspace{-3ex}
    \begin{align*}
            z^{(j)}(t) =& \lambda^{j}\Bigg(\lambda^{j+1}\bigg(\ldots\lambda^{(i-3)}\Big(\lambda^{(i-2)}\big(\lambda^{(i-1)}(x^{(i)}(t),x^{(i+1)}\ldots,x^{(N)}),
		x^{(i)}(t),x^{(i+1)},\ldots,x^{(N)}\big),\\
		&\lambda^{(i-1)}(x^{(i)}(t),x^{(i+1)},\ldots,x^{(N)}),
		x^{(i)}(t),x^{(i+1)},\ldots,x^{(N)}\Big),
		\ldots,x^{(i)}(t),x^{(i+1)},\ldots,x^{(N)}\bigg)\Bigg)
        \end{align*}

\section{Main Results}
\label{mainresults}
We state in this section (and later prove) our results for general $N$-timescale stochastic approximation algorithms as given in Section~\ref{NTS}.
\begin{theorem}
\label{main_thm}
We define for $1\leq j \leq N-2,$ $\lambda^{(j:N-1)}(x) \stackrel{\triangle}{=} \lambda^{(j)}\bigg(\lambda^{(j+1)}\Big(\ldots\lambda^{(N-2)}(\lambda^{(N-1)}(x),x),\ldots,x\Big),x\bigg)$.
 Under the assumptions $\textbf{(A:1)-(A:3)}$, \textbf{(B.N.i)}$_{1\leq i \leq N}$ and \textbf{(C.N.i)}$_{1\leq i \leq N}$, 
\begin{gather*}
        x_{n}^{(1)} \rightarrow x^{(1)}_*= \lambda^{(1:N-1)}(x^{(N)}_*),\\
        x_{n}^{(2)} \rightarrow x^{(2)}_*= \lambda^{(2:N-1)}(x^{(N)}_*),\\
        \longvdots{2em}\\
        x^{(N-1)}_{n} \rightarrow x^{(N-1)}_*= \lambda^{(N-1:N-1)}(x^{(N)}_*),\\
        x^{(N)}_{n} \rightarrow x^{(N)}_{*}.
\end{gather*}
\end{theorem}
We first prove that the $N$ iterates are convergent under the assumption that the $N$ iterates are stable. Specifically, we first prove a similar result (see Theorem~\ref{Thm_with_stability}) under the following assumption instead of \textbf{(B.N.i)}, $1\leq i \leq N$.
\begin{enumerate}
    \item [\textbf{(B.N.N+1)}] The following holds for the iterates:
    \(\sup_{n}\Big(\|x^{(1)}_{n}\| + \|x^{(2)}_{n}\| + \cdots + \|x^{(N)}_{n}\|\Big) < \infty \mbox{\quad a.s.}\)
\end{enumerate}
\begin{theorem}
\label{Thm_with_stability}
Under the assumptions $\textbf{(A:1)-(A:3)}$, $\textbf{(B.N.N+1)}$ and \textbf{(C.N.i)}$_{1\leq i \leq N}$, 
\begin{gather*}
        x_{n}^{(1)} \rightarrow x^{(1)}_*= \lambda^{(1:N-1)}(x^{(N)}_*),\;\;
        x_{n}^{(2)} \rightarrow x^{(2)}_*= \lambda^{(2:N-1)}(x^{(N)}_*),\\
        \longvdots{2em}\\
        x^{(N-1)}_{n} \rightarrow x^{(N-1)}_*=\lambda^{(N-1:N-1)}(x^{(N)}_*),\;\;
        x^{(N)}_{n} \rightarrow x^{(N)}_{*}.
\end{gather*}
\end{theorem}
Next, we shall show that the $N$ recursions are a.s. stable under the assumptions \textbf{(B.N.i)}$_{1\leq i \leq N}$. In other words, Assumption \textbf{(B.N.N+1)} holds under \textbf{(B.N.i)}$_{1\leq i \leq N}$.
\begin{theorem}
    \label{thm_without_stability}
    Under the assumptions $\textbf{(A:1)-(A:3)}$, \textbf{(B.N.i)}$_{1\leq i \leq N}$ and \textbf{(C.N.i)}$_{1\leq i \leq N}$, 
    \begin{gather*}
        \sup_{n}\Big(\|x^{(1)}_{n}\| + \|x^{(2)}_{n}\| + \cdots + \|x^{(N)}_{n}\|\Big) < \infty.
    \end{gather*}
\end{theorem}
Theorem \ref{thm_without_stability} with Theorem \ref{Thm_with_stability} imply that the recursions converge a.s. under the assumptions in Theorem \ref{main_thm}.
Finally, we consider the $N$-timescale recursions where each iterate contains a small additive perturbation term as follows:
\begin{align*}
    x^{(1)}_{n+1} &= x^{(1)}_{n} + \alpha^{(1)}_{n}(h^{(1)}(x^{(1)}_{n},x^{(2)}_{n},\ldots, x^{(N)}_{n}) + M^{(1)}_{n+1} + \varepsilon^{(1)}_{n}), \\
    x^{(2)}_{n+1} &= x^{(2)}_{n} + \alpha^{(2)}_{n}(h^{(2)}(x^{(1)}_{n},x^{(2)}_{n},\ldots,x^{(N)}_{n}) + M^{(2)}_{n+1}+\varepsilon^{(1)}_{n})),\\
    &\;\;\vdots \\
    x^{(N)}_{n+1} &= x^{(N)}_{n} + \alpha^{(N)}_{n}(h^{(N)}(x^{(1)}_{n},x^{(2)}_{n},\ldots,x^{(N)}_{n}) + M^{(N)}_{n+1}+\varepsilon^{(1)}_{n})).
\end{align*}

\begin{theorem}
\label{thm_with_pert}
Assume $\textbf{(A:1)-(A:3)}$, \textbf{(B.N.i)}$_{1\leq i \leq N}$ and \textbf{(C.N.i)}$_{1\leq i \leq N}$ hold. Further assume that $\sum_{i=1}^{N}\varepsilon^{(i)}_{n} \rightarrow 0$ as $n\rightarrow\infty$. Then $x^{(j)}_{n}, 1\leq j \leq N$ converges almost surely to the same solution as in Theorem \ref{main_thm}.
\end{theorem}

{\color{black} \subsection{Applications of the General Theory}
\label{sec:application_summary}

We now show how the abstract $N$-timescale convergence theorems yield concrete guarantees for three widely used reinforcement learning algorithms. In each case, we:  
(1) describe the objective being solved,  
(2) express the algorithm as a coupled stochastic recursion, and  
(3) map it to our general theory to obtain convergence results. A detailed convergence analysis of the algorithms can be found in Section~\ref{sec_app}.
\vspace{1ex}
\begin{enumerate}
    \item[\textbf{1.}]\noindent \textbf{GTD2 with momentum (policy evaluation).}
    The goal is to evaluate a fixed policy $\pi$ using off-policy data collected from a behaviour policy $\mu$. Let $v^\pi$ denote the value function of $\pi$. GTD2 seeks to find the best linear approximation $v_\theta(x) = \phi(x)^\top \theta$ to $v^\pi$, solving the projected Bellman equation while correcting the instability introduced by off-policy sampling. A heavy-ball momentum term is added to the recursion to improve empirical convergence rate.

    \noindent\textbf{Stochastic recursion.} The algorithm maintains three iterates, $v_t$, a fast auxiliary vector used for temporal-difference errors, $u_t$ a momentum term tracking smoothed value estimates, $\theta_t$, the main parameter approximating the value function. The updates evolve on three distinct timescales $(\xi_t, \beta_t, \varrho_t)$ satisfying $\xi_t > \beta_t > \varrho_t$, and follow linear recursions:
    \vspace{-2ex}
    \begin{align*}
        v_{t+1} &= v_t + \xi_t \big( h^{(1)}(v_t, u_t, \theta_t) + M^{(1)}_{t+1} \big), \\
        u_{t+1} &= u_t + \beta_t \big( h^{(2)}(v_t, u_t, \theta_t) + M^{(2)}_{t+1} \big), \\
        \theta_{t+1} &= \theta_t + \varrho_t \big( h^{(3)}(v_t, u_t, \theta_t) + M^{(3)}_{t+1} \big).
    \end{align*}
    The functions $h^{(i)}$ are affine in their arguments and the noise terms $M^{(i)}_{t+1}$ are martingale differences satisfying standard boundedness and moment conditions.

    \noindent\textbf{Application of theory.} Each $h^{(i)}$ is globally Lipschitz, and the joint drift satisfies the hierarchical contraction properties required by Assumptions~\textbf{(B)} and~\textbf{(C)}. The stepsize conditions match Assumption~\textbf{(A.2)}, and the feature/reward boundedness ensures that Assumption~\textbf{(A.1)} holds. Consequently, our Theorem~\ref{main_thm} implies:
    \[
        \theta_t \to \theta^* = -\bar A^{-1} \bar b \quad \text{almost surely},
    \]
    where $(\bar A, \bar b)$ are the expected GTD2 update matrices under the stationary distribution of $\mu$. The fast and intermediate iterates $(v_t, u_t)$ track their nested fixed points, with no additional stability constraints needed.
    \vspace{1ex}
    \item[\textbf{2.}] \noindent \textbf{Off-policy actor–critic (unconstrained policy optimization).} The objective is to maximise the expected discounted return $J(\pi)$ over a class of parameterised stochastic policies $\pi_{\theta}$, using samples drawn from an off-policy distribution. To stabilise training, the method uses a critic to estimate action-values (or advantages) and updates the actor via gradient ascent. An additional auxiliary parameter may be used to track target quantities (e.g., baselines or bootstrapped TD targets), yielding a three-timescale algorithm.

\noindent\textbf{Stochastic recursion.} The algorithm maintains three parameters, a fast critic state $v_t$ (e.g., value or TD-error vector), an intermediate auxiliary estimate $u_t$ (e.g., for baselines or GTD stability), a slow actor parameter $\theta_t$ optimizing $J(\theta)$ via gradient feedback. The updates evolve on three nested timescales $\xi_t > \beta_t > \alpha_t$, and follow the recursion:
\vspace{-2ex}
\begin{align*}
v_{t+1} &= v_t + \xi_t \big( h^{(1)}(v_t, u_t, \theta_t) + M^{(1)}_{t+1} \big), \\
u_{t+1} &= u_t + \beta_t \big( h^{(2)}(v_t, u_t, \theta_t) + M^{(2)}_{t+1} \big), \\
\theta_{t+1} &= \theta_t + \alpha_t \big( h^{(3)}(v_t, u_t, \theta_t) + M^{(3)}_{t+1} \big).
\end{align*}
The critic $v_t$ tracks a TD-style value estimate; $u_t$ may track an advantage baseline or auxiliary statistics; $\theta_t$ implements a policy gradient step using critic feedback.

\noindent\textbf{Application of theory.} All updates are Lipschitz and structured to obey a hierarchical contraction condition, so assumptions \textbf{(B)} and \textbf{(C)} hold. Boundedness of features and noise ensures assumptions \textbf{(A.1)} and \textbf{(A.3)}. The step-size hierarchy satisfies \textbf{(A.2)}. Then by Theorem~\ref{main_thm}:
\[v_t \to \lambda^{(1)}(u_t, \theta_t), u_t \to \lambda^{(2)}(\theta_t), \theta_t \to \theta^* \text{\;\;where\;\;} \nabla_\theta J(\theta^*) = 0, \text{ \;\; almost surely}.\]
Thus the actor–critic scheme converges to a stationary point of the return objective, even under off-policy sampling, with all internal auxiliary states tracking their fixed points.
\vspace{1ex}

    \item[\textbf{3.}]\noindent \textbf{Constrained actor–critic (policy optimization with constraints).} The goal is to solve a constrained reinforcement learning problem of the form
    \[
        \max_\pi J(\pi) \quad \text{subject to } C_i(\pi) \leq \nu_i,\quad 1\leq i\leq m,
    \]
    where both the objective and constraints depend on expectations over trajectories. The method uses Lagrangian dual ascent: an inner loop updates the actor and critic, and an outer loop adjusts the Lagrange multipliers $\gamma_t$.

    \noindent\textbf{Stochastic recursion.} The algorithm maintains a critic $(v_t)$ on the fastest timescale for value estimation, an actor $\theta_t$ on a slower timescale for primal ascent and dual variables $\gamma_t$ on the slowest timescale for constraint enforcement. These obey a three-timescale recursion of the form:
    \vspace{-2ex}
    \begin{align*}
        v_{t+1} &= v_t + \beta_t \big( h^{(1)}(v_t, \theta_t, \gamma_t) + M^{(1)}_{t+1} \big), \\
        \theta_{t+1} &= \theta_t + \alpha_t \big( h^{(2)}(v_t, \theta_t, \gamma_t) + M^{(2)}_{t+1} \big), \\
        \gamma_{t+1} &= \gamma_t + \eta_t \big( h^{(3)}(\theta_t, \gamma_t) + M^{(3)}_{t+1} \big),
    \end{align*}
    where $\beta_t > \alpha_t > \eta_t$ and the updates ascend the Lagrangian $\mathcal{L}(\theta, \gamma) = J(\theta) - \sum_i \gamma_i(C_i(\theta) - \nu_i)$.

    \noindent\textbf{Application of theory.} Under standard assumptions (bounded features, Lipschitz dynamics, monotone constraint violation), the triple recursion satisfies all the conditions of Theorem~\ref{main_thm}. Then, the critic tracks its fixed point for each $(\theta_t, \gamma_t)$, the actor follows a projected gradient flow for $\mathcal{L}(\cdot, \gamma_t)$, the dual variables ascend the dual function to enforce constraints. The iterates converge to a saddle point of the Lagrangian, yielding a primal-dual pair $(\theta^*, \gamma^*)$ such that:
    \[
        \nabla_\theta \mathcal{L}(\theta^*, \gamma^*) = 0, \qquad
        \gamma^*_i \geq 0, \quad C_i(\theta^*) \leq \nu_i, \quad
        \gamma^*_i(C_i(\theta^*) - \nu_i) = 0.
    \]
    This satisfies the KKT conditions for the constrained problem.

\end{enumerate}
}
{\color{black} \subsection{Proof Sketch of Theorem~\ref{main_thm}}
Before we proceed to the proof of the main result, we provide a brief sketch of the proof of Theorem~\ref{main_thm}:
\begin{enumerate}
    \item \textbf{Reduction to a stability--convergence decomposition.}
    The argument splits naturally into two logically independent parts.  
    First, assuming that the entire iterate vector
    \((x_n^{(1)},x_n^{(2)},\dots,x_n^{(N)})\)
    remains almost surely bounded, we demonstrate that the $N$ coupled stochastic‐approximation (SA) recursions converge to the prescribed hierarchical fixed point described in the theorem statement.  
    Second, we remove that boundedness assumption by proving the required \emph{a.s.} stability of all coordinates directly from the drift conditions in Assumptions~\textbf{(B.N.i)}\(_{1\le i\le N}\).  
    The two ingredients, taken together, yield the full claim.

    \item \textbf{ODE limit for the fastest scale under boundedness.}
    Fix the fastest stepsize sequence \(\{\alpha_n^{(1)}\}\).
    After dividing the $j$th recursion by the common factor \(\alpha_n^{(1)}\) we observe that, for every \(2\le j\le N\), the ratio \(\alpha_n^{(j)}/\alpha_n^{(1)}\) decays to zero by Assumption~\textbf{(A:3)}(iii).  
    Thus all coordinates except \(x^{(1)}\) evolve on an \emph{infinitely slower} clock and act, from the perspective of the \(x^{(1)}\) dynamics, as \emph{frozen} external parameters.  
    Invoking the standard ODE method (Benaïm–Borkar fram\-e\-work) with the “third extension’’ of Chap.~2 in \cite{BorkarBook} shows that the joint process tracks the flow of the differential inclusion
    \[
    \dot{x}^{(1)}(t)=h^{(1)}\bigl(x^{(1)}(t),x^{(2)}(t),\dots,x^{(N)}(t)\bigr),\quad
    \dot{x}^{(j)}(t)=0\;(2\le j\le N),
    \]
    whose internally chain–transitive (ICT) invariant set is
    \(
    \bigl\{(\lambda^{(1)}(x^{(2)},\dots,x^{(N)}),x^{(2)},\dots,x^{(N)})\bigr\}.
    \)
    Consequently
    \(
    x_n^{(1)}-\lambda^{(1)}(x_n^{(2)},\dots,x_n^{(N)})\to0
    \)
    almost surely as \(n\to\infty\).

    \item \textbf{Cascade down the timescales.}
    We repeat the identical argument on the \emph{next} stepsize sequence \(\{\alpha_n^{(2)}\}\) while treating the remaining coordinates \((x^{(3)},\dots,x^{(N)})\) as quasi‐static.  
    Induction shows that for each \(1\le k\le N-1\)
    \[
    x_n^{(k)}-
    \lambda^{(k:N-1)}\!\bigl(x_n^{(N)}\bigr)\;\longrightarrow\;0
    \quad\text{a.s.}
    \]
    where the iterated map \(\lambda^{(k:N-1)}\) is the nested fixed‐point operator from the theorem statement.

    \item \textbf{ODE analysis on the slowest scale.}
    Introduce a continuous‐time interpolation of the slowest iterate \(x^{(N)}_n\) and analyse the martingale‐error decomposition exactly as in Chap.~2, Lemma 2 of \cite{BorkarBook}.  
    Gronwall’s inequality bounds the difference between the interpolated trajectory and the solution of the limit ODE
    \[
    \dot{x}^{(N)}(t)=h^{(N)}\!\Bigl(
    \lambda^{(1:N-1)}\!\bigl(x^{(N)}(t)\bigr),\;
    \dots,\;
    \lambda^{(N-1)}\!\bigl(x^{(N)}(t)\bigr),\;
    x^{(N)}(t)
    \Bigr),
    \]
    whose globally asymptotically stable equilibrium is \(x_*^{(N)}\) by Assumption~\textbf{(C.N.N)}.  
    Hence \(x^{(N)}_n\to x_*^{(N)}\) and, by Step 2, every faster coordinate converges to \(\lambda^{(k:N-1)}(x_*^{(N)})\), completing the convergence proof \emph{conditional on stability}.  This yields Theorem~\ref{Thm_with_stability}.

    \item \textbf{Bounding the iterates: a recursive scaling argument.}
    To prove almost‐sure boundedness we compare the raw iterates with appropriately scaled versions that satisfy rescaled SA recursions.  
    For the fastest scale one rescales by \(r(n)\doteq\max\bigl(1,\|\bar{X}^{(1)}(T_n)\|\bigr)\); Lipschitz continuity of \(h^{(1)}\) and the drift condition \textbf{(B.N.1)} force the scaled trajectory to contract whenever it exits a suitable multiple of the slower‐scale coordinates.  
    Formally, Lemma~\ref{fast_main} shows that
    \(
    \|x^{(1)}_n\|\le K_1^*\sum_{j=2}^N\|x^{(j)}_n\|.
    \)

    The identical construction applied to the $l$th timescale—now with the first \(l-1\) coordinates treated as instantaneous and the rest as quasi‐static—yields, by induction over \(l=2,\dots,N-1\),
    \[
    \|x^{(l)}_n\|\le K_l^*\sum_{j=l+1}^N\|x^{(j)}_n\|.
    \]
    Finally, for the slowest iterate, one shows by contradiction that its norm cannot diverge; otherwise the scaled recursion would violate the one‐step contraction inequality in Lemma~\ref{slow_main}.  As a result all coordinates remain a.s.\ bounded, validating assumption \textbf{(B.N.N+1)} and establishing Theorem~\ref{thm_without_stability}.

    \item \textbf{Combining stability and convergence.}
    The boundedness of every coordinate proved in Step 5 feeds into Step 1–Step 3, thereby proving Theorem~\ref{main_thm} in full generality.

\end{enumerate}}
\vspace{5ex}
\begin{remark}
    Since the additional error terms are $o(1)$, their contribution is asymptotically negligible. See arguments in the third extension of (Chapter 2, pp. 17 of \cite{BorkarBook}) that handles this case for one-timescale iterates. Using similar arguments along with our analysis, this extension can be easily obtained. 
\end{remark}

\section{Proof of the Main Results}
\label{proof}
\subsection{Showing convergence by assuming Stability (Theorem \ref{Thm_with_stability})}

We start by characterizing the set to which the iterate-vector $(x^{(1)}_{n},x^{(2)}_{n},\ldots,x^{(N)}_{n})$ converges. Extending the arguments in Lemma 1, Chapter 6 of \cite{BorkarBook},
we first consider the timescale of $\{\alpha^{(1)}_{n}\}$ and rewrite the $N$ iterates as follows:
\begin{gather*}
        x^{(1)}_{n+1} = x^{(1)}_{n} + \alpha_{n}^{(1)}\bigg(h^{(1)}\Big(x^{(1)}_{n},x^{(2)}_{n},\ldots,x^{(N)}_{n}\Big) + M_{n+1}^{(1)}\bigg),\\
        x^{(2)}_{n+1} = x^{(2)}_{n} + \alpha_{n}^{(1)}\bigg(\epsilon_{n}^{(2)} + \bar{M}_{n+1}^{(2)}\bigg),\\
        x^{(3)}_{n+1} = x^{(3)}_{n} + \alpha_{n}^{(1)}\bigg(\epsilon_{n}^{(3)} + \bar{M}_{n+1}^{(3)}\bigg),\\
        \longvdots{2em}\\
        x^{(N)}_{n+1} = x^{(N)}_{n} + \alpha_{n}^{(1)}\bigg(\epsilon_{n}^{(N)} + \bar{M}_{n+1}^{(N)}\bigg),
\end{gather*}
where, we define $\forall 2\leq j \leq N,$
    $\epsilon_{n}^{(j)} = \frac{\alpha^{(j)}_{n}}{\alpha^{(1)}_{n}}h^{(j)}\Big(x^{(1)}_{n},x^{(2)}_{n},\ldots,x^{(N)}_{n}\Big), \mbox{\quad} \bar{M}_{n+1}^{(j)} = \frac{\alpha^{(j)}_{n}}{\alpha^{(1)}_{n}} M_{n+1}^{(j)}.$  

\noindent From assumption \textbf{(A:3)}(iii), $\frac{\alpha^{(j)}_{n}}{\alpha^{(1)}_{n}} \rightarrow 0$, as $n\rightarrow\infty$, $ \forall 2\leq j \leq N$ and thus $\epsilon_{n}^{(j)} = o(1)$ a.s. Using now the third extension from Chapter-2 of \cite{BorkarBook}, $\Big(x^{(1)}_{n},x^{(2)}_{n},\ldots,x^{(N)}_{n}\Big)$ converges to an internally chain transitive invariant set of the ODE
\begin{gather*}
    \dot{x}^{(1)}(t) = h^{(1)}\Big(x^{(1)}(t), x^{(2)}(t), \ldots, x^{(N)}(t)\Big),\quad
    \dot{x}^{(2)}(t) = 0,\quad
    \dot{x}^{(3)}(t) = 0,\quad\ldots,
    \dot{x}^{(N)}(t) = 0.
\end{gather*}
For initial conditions ${x}^{(j)} \in \mathbb{R}^{d_j}, 1 \leq j \leq N$, the internally chain transitive invariant set of the above ODE system is $\Big\{\Big(\lambda^{(1)}\Big(x^{(2)},x^{(3)},\ldots,x^{(N)}\Big),x^{(2)},x^{(3)},\ldots,x^{(N)}\Big)\Big\}$. Therefore,
\begin{gather}
    \Big(x^{(1)}_{n}, x^{(2)}_{n}, \ldots, x^{(N)}_{n}\Big) \rightarrow 
    \bigg\{\big(\lambda^{(1)}\big(x^{(2)},x^{(3)},\ldots,x^{(N)}\big),x^{(2)},x^{(3)},\ldots,x^{(N)}\big),x^{(j)}\in \mathbb{R}^{d_j}, 2\leq j \leq N\bigg\}.
    \label{1_ts}
\end{gather}
Next, we consider the timescale $\alpha^{(m)}_n$, $2 \leq m \leq N-1$,
\begin{gather*}
    x^{(m)}_{n+1} = x^{(m)}_{n} + \alpha_{n}^{(m)}\bigg(h^{(m)}\Big(x^{(1)}_{n},x^{(2)}_{n},\ldots,x^{(N)}_{n}\Big) + M_{n+1}^{(m)}\bigg),\\
    x^{(m+1)}_{n+1} = x^{(m+1)}_{n} + \alpha_{n}^{(m)}\bigg(\epsilon_{n}^{(m+1)} + \bar{M}_{n+1}^{(m+1)}\bigg),\\
    \longvdots{2em}\\
    x^{(N)}_{n+1} = x^{(N)}_{n} + \alpha_{n}^{(m)}\bigg(\epsilon_{n}^{(N)} + \bar{M}_{n+1}^{(N)}\bigg),
\end{gather*}
where we re-define for $m+1 \leq j \leq N$ (note the abuse of notation)
\begin{gather*}
    \epsilon_{n}^{(j)} = \frac{\alpha^{(j)}_{n}}{\alpha^{(m)}_{n}}h^{(j)}\Big(x^{(1)}_{n},x^{(2)}_{n},\ldots,x^{(N)}_{n}\Big), \mbox{\quad} \bar{M}_{n+1}^{(j)} = \frac{\alpha^{(j)}_{n}}{\alpha^{(m)}_{n}} M_{n+1}^{(j)}.
\end{gather*}

From assumption $\textbf{(A:3)} (iii)$, $\frac{\alpha^{(j)}_{n}}{\alpha^{(m)}_{n}} \rightarrow 0$ as $n\rightarrow \infty$ and therefore $\epsilon_{n}^{(j)} = o(1)$ $\forall$ $m+1\leq j \leq N$. Using the third extension from Chapter 2 of \cite{BorkarBook}, $\big(x^{(m+1)}_{n},x^{(m+2)}_{n},\ldots,x^{(N)}_{n}\big)$ converges to an internally chain transitive invariant set of the ODE
\begin{gather*}
    \dot{x}^{(m+1)}(t) = h^{(m+1)}\Big(x^{(1)}(t), x^{(2)}(t), \ldots, x^{(N)}(t)\Big),\;\;
    \dot{x}^{(m+2)}(t) = 0,\;\;\ldots,\;\;
    \dot{x}^{(N)}(t) = 0.
\end{gather*}
For initial conditions $x^{(j)} \in \mathbb{R}^{d_j},$ $m \leq j \leq N$, the internally chain transitive invariant set of the ODE system is $\Big\{\Big(\lambda^{(m)}\Big(x^{(m+1)},\ldots,x^{(N)}\Big),x^{(m+1)},x^{(m+2)},\ldots,x^{(N)}\Big)\Big\}$. Therefore,
\begin{gather}
    \big(x^{(m)}_{n}, x^{(m+1)}_{n}, \ldots, x^{(N)}_{n}\big) \rightarrow \nonumber
    \big\{\big(\lambda^{(m)}\big(x^{(m+1)},x^{(m+2)},\ldots,x^{(N)}\big),x^{(m+1)},x^{(m+2)},\ldots,x^{(N)}\big),x^{(j)}\in \mathbb{R}^{d_j}, m+1\leq j \leq N\big\}.
    \label{m_ts}
\end{gather} 
We now infer the following from \eqref{m_ts}:
\begin{gather}
    \Big(x^{(1)}_{n},x^{(2)}_{n},\ldots,x^{(N)}_{n}\Big)\rightarrow
    \Bigg\{\bigg(\lambda^{(1)}\bigg(\lambda^{(2)}\Big(\ldots\lambda^{(N-2)}(\lambda^{(N-1)}(x^{(N)}),
    x^{(N)}),\ldots,x^{(N)}\Big),x^{(N)}\bigg)
    \mbox{{,}}\nonumber\\
    \lambda^{(2)}\Big(\ldots\lambda^{(N-2)}\big(\lambda^{(N-1)}(x^{(N)}),x^{(N)}\big),\ldots,x^{(N)}\Big)\mbox{{,}}\ldots
    \ldots\mbox{{,}}\lambda^{(N)}(x^{(N)})\mbox{{,}}x^{(N)}\bigg) \bigg|\; x^{(N)}\in\mathbb{R}^{d_N}\Bigg\}\mbox{\quad a.s.}   
    \label{eq_10}
\end{gather}
Recall now that
\begin{gather*}
    \lambda^{(j:N-1)}(x) \stackrel{\triangle}{=}  \lambda^{(j)}\bigg(\lambda^{(j+1)}\Big(\ldots\lambda^{(N-2)}(\lambda^{(N-1)}(x),
    x),\ldots,x\Big),x\bigg),
    1\leq j \leq N-2.
\end{gather*}
Then, $\forall 1\leq k \leq N-1$, 
\begin{equation}
    \label{eq_8}
    \big\|\lambda^{(k:N-1)}(x_{n}^{(N)}) - x_{n}^{(k)}\big\| \rightarrow 0 \mbox{ a.s. as } n\rightarrow\infty.
\end{equation}
Finally, we consider the slowest timescale, i.e., the one corresponding to  $\{\alpha^{(N)}_n\}$. We define the piece-wise linear continuous interpolation of the iterates $x^{(N)}_{n}$ as follows:
\begin{gather*}
    \bar{x}^{(N)}\big(t(n)\big) = {x}^{(N)}_{n},\\
    \bar{x}^{(N)}(t) = {x}^{(N)}_{n} + \big({x}^{(N)}_{n+1} - {x}^{(N)}_{n}\big)\frac{t-t(n)}{t(n+1) - t(n)}\;,\quad
    t \in [t(n),t(n+1)],
\end{gather*}
where, $t(n) = \sum_{m=0}^{n-1}\alpha^{(N)}_{m}, n\geq 1$ with $t(0)=0$. Also, let ${x}^{(N),s}(t), t\geq s$, denote the unique solution of the below ODE starting at $s\in \mathbb{R}:$
\[{\dot{x}}^{(N),s}(t) = h^{(N)}\bigg(\lambda^{(1:N-1)}\big(x^{(N),s}(t)\big),
    \lambda^{(2:N-1)}\big(x^{(N),s}(t)\big),\\\ldots,\]
    \[\lambda^{(N-2:N-1)}\big(x^{(N),s}(t)\big),
    \lambda^{(N-1)}\big(x^{(N),s}(t)\big), x^{(n),s}(t) \bigg), \mbox{ } t\geq s,\]
with ${x}^{(N),s}(s) = \bar{x}^{(N)}(s)$. 
Let 
\(\displaystyle \psi^{(N)}_{n} = \sum_{i=0}^{n-1}\alpha^{(N)}_{i}M_{i+1}^{(N)}, \mbox{ }n\geq1.\)
It is easy to see that $\{\psi^{(N)}_n,{\cal F}_n\}$ is a zero-mean square integrable martingale. Further,
\[\sum_{n\geq 0} E[|\ \psi^{(N)}_{n+1}-\psi^{(N)}_n\|^2 | \mathcal{F}_n] 
= \sum_{n\geq 0} (\alpha^{(N)}_n)^2 E[\|M^{(N)}_{n+1}\|^2|\mathcal{F}_n] <\infty,
\]
a.s., from {\bf (A:2)}, {\bf (A:3)} and {\bf (B.N.N+1)}. From the martingale convergence
theorem, $\{\psi^{(N)}_n,\mathcal{F}_n\}$ is an almost surely convergent martingale sequence. Let $[t] = \max\{t(n): t(n) \leq t\}$, $t\geq 0$. Then for $n,m \geq 0$.
\begin{gather*}
    \bar{x}^{(N)}(t(n+m)) = \bar{x}^{(N)}(t(n)) + \int_{t(n)}^{t(n+m)} h^{(N)}\bigg(\lambda^{(1:N-1)}\big(\bar{x}^{(N)}(t)\big),
    \lambda^{(2:N-1)}\big(\bar{x}^{(N)}(t)\big),\\\ldots,\lambda^{(N-2:N-1)}\big(\bar{x}^{(N)}(t)\big),
    \lambda^{(N-1)}\big(\bar{x}^{(N)}(t)\big), \bar{x}^{(N)}(t)\bigg) + I + II + III,
\end{gather*}
where,
\begin{gather*}
    I = \Bigg\|\int_{t(n)}^{t(n+m)}\Bigg(h^{(N)}\Big(\lambda^{(1:N-1)}\big(\bar{x}^{N}([t])\big),\lambda^{(2:N-1)}\big(\bar{x}^{N}([t])\big),\ldots,
    \lambda^{(N-1)}\big(\bar{x}^{N}([t])\big),\bar{x}^{N}([t])\Big)\\ - h^{(N)}\Big(\lambda^{(1:N-1)}\big(\bar{x}^{N}(t)\big),
    \lambda^{(2:N-1)}\big(\bar{x}^{N}(t)\big),
    \ldots,\lambda^{(N-1)}\big(\bar{x}^{N}(t)\big),\bar{x}^{N}(t)\Big)\Bigg)dt \Bigg\|,\\
    II = \Bigg\|\sum_{j=0}^{m-1}\alpha^{(N)}_{j}\bigg(h^{(N)}\big(x^{(1)}_{n+j},x^{(2)}_{n+j},\ldots,x^{(N)}_{n+j}\big)
    \\- h^{(N)}\big(\lambda^{(1:N-1)}
    \big(\bar{x}^{N}_{n+j}\big),\lambda^{(2:N-1)}\big(\bar{x}^{N}_{n+j}\big),
    \ldots,\lambda^{(N-1)}\big(\bar{x}^{N}_{n+j}\big),\bar{x}^{N}_{n+j}\big)\bigg)\Bigg\|, \\
    III = \Big\|\Big(\psi^{(N)}_{n+m+1} - \psi^{(N)}_{n}\Big)\big\|.
\end{gather*}
Now, as in Lemma 2 of Chapter 2 in \cite{BorkarBook}, using Gronwall inequality we have:
\begin{gather*}
    \sup_{t\in[s,s+T]}\|\bar{x}^{N}(t) - {x}^{(N),s}(t)\| \leq K_{T}\Big(I+II+III\Big),
\end{gather*}
where, $K_{T}>0$ is a constant that depends on $T$. We next show that all the three terms on the RHS above go to $0$ as $n\rightarrow\infty$. Using Lipschitz continuity of $h^{(N)},$ we have
\begin{gather*}
    \|I\| \leq \int_{t(n)}^{t(n+m)}L\Bigg(\Big\|\lambda^{(1:N-1)}\big(\bar{x}^{(N)}([t])\big) - \lambda^{(1:N-1)}\big(\bar{x}^{(N)}(t)\big)\Big\| \\
    + \Big\|\lambda^{(2:N-1)}\big(\bar{x}^{(N)}([t])\big) - \lambda^{(2:N-1)}\big(\bar{x}^{(N)}(t)\big)\Big\| + \cdots
    + \Big\|\bar{x}^{(N)}([t]) - \bar{x}^{(N)}(t)\Big\|\Bigg)\\
    \leq \sum_{j=0}^{m-1}K_{L}\int_{t(n+j)}^{t(n+j+1)}\big\|\bar{x}^{(N)}(t(n+j)) - \bar{x}^{(N)}(t)\big\|,
\end{gather*}
where $K_{L} = L(L^{N-1}+L^{N-2}+\cdots+L+1)$ and we have used the fact that $\lambda^{(j:N)}$ is Lipschitz $\forall 1\leq j \leq N-2$ and $[t] = t(n+j)$ for $t\in[t(n+j),t(n+j+1)], \forall$ $ 1\leq j \leq m-1$. Now, 
\begin{gather*}
    \|\bar{x}^{(N)}(t(n+j)) - \bar{x}^{(N)}(t)\|\leq \frac{\Big(\|x^{(N)}_{n+j+1} - x^{(N)}_{n+j\|}\Big)t}{t(n+j+1)-t(n+j)}.
\end{gather*}
Therefore, using A:3 (ii)
\begin{align*}
    \|I\| \leq \sum_{j=0}^{m-1}\mathcal{O}(\alpha^{(N)}_{n+j})\int_{t(n+j)}^{t(n+j+1)} \frac{t}{t(n+j+1)-t(n+j)}dt\leq \sum_{j>n}\mathcal{O}(\alpha^{(N)}_{j})^2 \rightarrow 0 \mbox{ as } n \rightarrow \infty. 
\end{align*}
Next, using Lipschitz continuity of $h^{N}(\cdot)$,
\begin{gather*}
    \|II\| \leq \Big\|\sum_{j=0}^{m-1} \alpha^{(N)}_{j}L\sum_{k=1}^{N-1}\Big(\|x^{(k)}_{n+j} - \lambda^{(k:N-1)}(x^{(N)}_{n+j})\|\Big) \Big\|\\
    \leq L \sum_{k=1}^{N}\mathcal{O}\Big(\sup_{j\geq n}\Big\|x^{(k)}_{j} - \lambda^{(k:N-1)}(x^{(N)}_{j})\Big\|\Big)
    \rightarrow 0 \mbox{ a.s. as } n\rightarrow\infty \mbox{ \big(from \eqref{eq_8}}\big)\\
\end{gather*}
Finally, from the foregoing, $\|III\| \leq \sup_{j>n}\|\psi^{(N)}_{j} - \psi^{(N)}_{n}\| \rightarrow 0$ almost surely as $n\rightarrow\infty$. Now, using arguments as in Lemma 1 of Chapter 2 in \cite{BorkarBook}, it can be shown that
\begin{gather*}
    \lim_{s\rightarrow\infty} \sup_{t \in [s,s+T]} \|\bar{x}^{(N)}(t) - x^{(N),s}(t)\| = 0 \mbox{ a.s. }
\end{gather*}
Since, $x^{(N)}_{*}$ is a globally asymptotically stable equilibrium with $i=N$, using Theorem 2 of Chapter 2 in \cite{BorkarBook}, we get $x^{(N)}_{n} \rightarrow x^{(N)}_{*}$ a.s. as $n\rightarrow\infty$. Combining this with \eqref{eq_10}, proves Theorem \ref{Thm_with_stability}.


\subsection{Showing Stability of the recursions 
(Theorem \ref{thm_without_stability})}

We begin with the fastest timescale recursion governed by the step-size sequence $\{\alpha^{(1)}_{n}\}$. We first state some of the notations and definitions used.

\begin{enumerate}
    \item[\bf{(D1)}] Let
    \[t^{(1)}(n) = \sum_{i=0}^{n-1}\alpha^{(1)}_i, n\geq 1, \mbox{ with } t^{(1)}(0) = 0.\]
    Let $X_k \stackrel{\triangle}{=} (x^{(1)}_k,x^{(2)}_k,\ldots,x^{(N)}_k)$, $k\geq0$. For $t \in [t^{(1)}(n), t^{(1)}(n+1)]$,
    \[\bar{X}^{(1)}(t) = X_{n} + \left(X_{n+1} - X_{n}\right)\frac{t-t^{(1)}(n)}{t^{(1)}(n+1) - t^{(1)}(n)}.\]
    
    \item[\bf{(D2)}]Given $t^{(1)}(n), n\geq 0$ and a constant $T>0$ define
    \[T_0 = 0,\;\;\text{and } T_n = \min(t^{(1)}(m):t^{(1)}(m)\geq T_{n-1}+T), n\geq 1.\]
    One can find a subsequence $\{m(n)\}$ such that $T_n = t^{(1)}(m(n))$ $\forall n$ and $m(n) \rightarrow \infty$ as $n\rightarrow \infty$.
    
    \item[\bf{(D3)}] Define a sequence $r(n)\geq 1$, $\forall n$ as follows:
    \[r(n) = \max(r(n-1),\|\bar{X}^{(1)}(T_n)\|,1).\]
    
    \item[\bf{(D4)}] Define the scaled iterates (obtained from the recursion above) for $m(n)\leq k\leq m(n+1)-1$ as:
    \[\hat{x}^{(1)}_{1} \stackrel{\triangle}{=} \frac{x_{k}^{(1)}}{r(n)}, \mbox{ }\hat{x}^{(2)}_{k} \stackrel{\triangle}{=} \frac{x^{(2)}_{k}}{r(n)},\ldots, \hat{x}^{(N)}_{k} \stackrel{\triangle}{=} \frac{x_{k}^{(N)}}{r(n)}.\]
Further,    
    \[\hat{x}^{(1)}_{k+1} = \hat{x}^{(1)}_{k} + \alpha^{(1)}_k\left(\frac{h^{(1)}(c\hat{x}^{(1)}_{k},\ldots,c\hat{x}^{(N)}_{k})}{c} + \hat{M}_{k+1}^{(1)}\right),\]
    \[\hat{x}^{(2)}_{k+1} = \hat{x}^{(2)}_{k} + \alpha^{(1)}_{k}\left(\epsilon_{k}^{(2)} + \hat{M}_{k+1}^{(2)}\right)\]
    \[\longvdots{2em}\]
    \[\hat{x}^{(N)}_{k+1} = \hat{x}^{(N)}_{k} + \alpha^{(1)}_{k} \left(\epsilon_{k}^{(N)} + \hat{M}_{k+1}^{(3)}\right)\]
    where, $c = r(n)$, and
    \[\epsilon_{k}^{(j)} = \frac{\alpha^{(j)}_k}{\alpha^{(1)}_k}\frac{h^{(j)}(c\hat{x}^{(1)}_{k},\ldots,c\hat{x}^{(N)}_{k})}{c}, \mbox{ } 2\leq j\leq N.\]
    \[\hat{M}_{k+1}^{(1)} = \frac{M_{k+1}^{(1)}}{r(n)}, \hat{M}_{k+1}^{(j)} = \frac{\alpha^{(j)}_k}{\alpha^{(1)}_k}\frac{M_{k+1}^{(j)}}{r(n)}, 2\leq j\leq N.\]
    
    \item[\bf{(D5)}] Next we define the linearly interpolated trajectory for the scaled iterates $\hat{X}^{(1)}_n=(\hat{x}^{(1)}_{n},\hat{x}^{(2)}_{n},\ldots,
    \hat{x}^{(N)}_{n})$ as follows:
    \begin{gather*}
        \hat{X}^{(1)}(t) = \hat{X}^{(1)}_n + (\hat{X}^{(1)}_{n+1} - \hat{X}^{(1)}_n)\frac{t - t^{(1)}(n)}{t^{(1)}(n+1) - t^{(1)}(n)},\;\;
    t \in [t^{(1)}(n), t^{(1)}(n+1)].
    \end{gather*}
    
    \item[\bf{(D6)}] Let \(X^{(1)}_{n}(t) \stackrel{\triangle}{=} (x^{(1)}_{n}(t),x^{(2)}_{n}(t),\ldots,x^{(N)}_{n}(t)), t\in [T_n,T_{n+1}]\) denote the trajectory of the ODE:
    \[\dot{x}^{(1)}(t) = h_{r(n)}(x^{(1)}(t),x^{(2)}(t),\ldots,x^{(N)}(t)),\;\; \dot{x}^{(2)}(t) = 0,\;\;\ldots,\;\; \dot{x}^{(N)}(t) = 0,\]
    with $x^{(j)}_{n}(T_n) = \hat{x}^{(j)}(T_n),$ $\forall 1\leq j \leq N$.
\end{enumerate}

We state a lemma for an ODE with external inputs. Let $x_{c}^{(1),x^{(2)}(t),\ldots,x^{(N)}(t)}$ and $x_{\infty}^{(1),x^{(2)}(t),\ldots,x^{(N)}(t)}$ denote the trajectories of the following ODEs:
\vspace{-2ex}
\begin{align*}
    \dot{x}^{(1)}(t) &= h^{(1)}_c\Big(x^{(1)}(t),x^{(2)}(t),\ldots, x^{(N)}(t)\Big),\\
    \dot{x}^{(1)}(t) &= h^{(1)}_{\infty}\Big(x^{(1)}(t),x^{(2)}(t),\ldots, x^{(N)}(t)\Big),
\end{align*}
respectively, with initial condition $x^{(1)}\in \mathbb{R}^{d_1}$ and the external inputs $x^{(j)}(t) \in \mathbb{R}^{d_j},$ $\forall 2\leq j \leq N$. Let $B^{j}(x^{(j)},r) \triangleq \{q\in\mathbb{R}^{d_j}\Big|\|q-x^{(j)}\|<r\}$ denote a ball of radius $r$ centred at $x^{(j)}$.
\begin{lemma}
\label{fastest_aux}
    Let $x^{(1)} \in B^{1}(0,1) \subset \mathbb{R}^{d_1},x^{(j)} \in W^{(j)}\subset\mathbb{R}^{d_j}, \forall 2\leq j \leq N$ and let \textbf{(B.N.1)} hold. Then given $\epsilon>0, \exists c_{\epsilon}\geq 1, r_{\epsilon}>0$ and $T_{\epsilon}>0$ such that for any external inputs satisfying \(\tilde{x}^{(j)}(s) \in B^{j}(x^{(j)},r_{\epsilon})\), $\forall s \in [0,T], \forall c > c_{\epsilon},t\geq T_{\epsilon}$, 
        \[\|x_c^{(1),\tilde{x}^{(2)}(t),\ldots,\tilde{x}^{(N)}(s)}(t,x) - \lambda_{\infty}^{(1)}(x^{(2)},\ldots,x^{(N)})\| \leq 2\epsilon.\]
\end{lemma}
The next lemma uses the convergence result of $N$ scale iterates under the stability assumption of \textbf{(B.N.N+1)} and shows that the scaled iterates defined in \textbf{(D4)} converge.   
\begin{lemma}
    \label{fast_lemma}
    Under \textbf{(A:1)}-\textbf{(A:3)},
    \begin{enumerate}
        \item[(i)] For \(0\leq k \leq m(n+1)-m(n)\),
        \(\|\hat{X}^{(1)}(t(m(n)+k))\| \leq K^{(1)}\)  a.s. for some constant $K^{(1)}>0.$
        
        \item[(ii)] \(\lim_{n\rightarrow\infty}\sup_{t\in[T_n,T_{n+1}]}
        \|\hat{X}^{(1)}(t) - X^{(1)}_n(t)\| = 0 \mbox{ a.s. }\)
    \end{enumerate}
\end{lemma}
The proof of the above two Lemmas follows in a similar manner as that of Lemmas 5 and 6, respectively, of \cite{chandru-SB}. In particular, Lemma \ref{fast_lemma}(i) shows that along the timescale of $\{\alpha_n^{(1)}\}$, between instants $T_n$ and $T_{n+1}$, the norm of the scaled iterate can grow at most by a factor $K^{(1)}$ starting from $B^{1}(0,1)$. Next, Lemma \ref{fast_lemma}(ii) shows that the scaled iterate asymptotically tracks the ODE defined in \textbf{(D6)}.
The next lemma bounds $\|x^{(1)}_n\|$ in terms of $\|x^{(2)}_n\|,\ldots, \|x^{(N)}_n\|$. We define the linearly interpolated trajectories of the $N$ iterates as follows: $\bar{x}^{(j)}(t(n))={x}^{(j)}$ and $ \forall t \in [t(n), t(n+1)]$,
    \[\bar{x}^{(j)}(t) = x^{(j)}_{n} + \left(x^{(j)}_{n+1} - x^{(j)}_{n}\right)\frac{t-t^{(1)}(n)}{t^{(1)}(n+1) - t^{(1)}(n)}.\]

\begin{lemma}
\label{fast_main}
Under assumptions \textbf{(A:1)}-\textbf{(A:3)}, \textbf{(B.N.i)}$_{1\leq i\leq N}$ and \textbf{(C.N.i)}$_{1\leq i\leq N}$,
    \begin{enumerate}
        \item[(i)] For $n$ large, and $T = T_{\frac{1}{4}}$, if 
        \(\|\bar{x}^{(1)}(T_n)\| > C_{1}( \|\bar{x}^{(2)}(T_n)\| + \cdots + \|\bar{x}^{(N)}(T_n)\|)\), for some $C_1>0$
        then \(\|\bar{x}^{(1)}(T_{n+1})\|\leq \frac{3}{4}\|\bar{x}^{(1)}(T_n)\|\).
        
        \item[(ii)] \(\|\bar{x}^{(1)}(T_n)\| \leq C_{1}^{*}(\|\bar{x}^{(2)}(T_n)\|+\cdots+ \|\bar{x}^{(N)}(T_n)\|)\) a.s. for some $C_{1}^*\geq1$.
        
        \item[(iii)] \(\|x^{(1)}_n\| \leq K_{1}^*\big(\|x^{(2)}_n\|+\cdots+\|x^{(N)}_n\|\big), \mbox{ for some } K_{1}^*>0.\)
    \end{enumerate}
\end{lemma}
\emph{Proof.} 
(i)
We have \(\|\bar{x}^{(1)}(T_n)\| > C_{1}(1 + \|\bar{x}^{(2)}(T_n)\| + \cdots + \|\bar{x}^{(N)}(T_n)\|).\) Since, \(r(n) = \max(r(n-1),\|\bar{X}^{(1)}(T_n)\|,1)\),  \(r(n)\geq\|\bar{X}^{(1)}(T_n)\|\). Therefore, $r(n)\geq C_1$. Next we show \(\|\hat{x}^{(j)}(T_n)\|<\frac{1}{C_1},\;\forall\; 2\leq j \leq N.\)
    \begin{align*}
            \text{For $p\geq1$, } \;\; \|\hat{x}^{(j)}(T_n)\|_p = \frac{\|\bar{x}^{(j)}(T_n)\|_p}{r(n)}\leq\frac{\|\bar{x}^{(j)}(T_n)\|_p}{\|\bar{X}^{(1)}(T_n)\|_p}
            =\frac{\|\bar{x}^{(j)}(T_n)\|_p}{\Big(\|\bar{x}^{(1)}(T_n)\|_p^p +\cdots+ \|\bar{x}^{(N)}(T_n)\|_p^p \Big)^{\frac{1}{p}}}.
    \end{align*}
    Since, \(\|\bar{x}^{(1)}(T_n)\|_p \geq C_1(1 + \|\bar{x}^{(2)}(T_n)\|_p +\cdots+\|\bar{x}^{(N)}(T_n)\|_p)\),
    \begin{align*}
            \|\bar{x}^{(1)}(T_n)\|_p^p \geq C_1^p\Big(\|\bar{x}^{(2)}(T_n)\|_p +\cdots+ \|\bar{x}^{(N)}(T_n)\|_p\Big)^p\geq C_1^p\Big(\|\bar{x}^{(2)}(T_n)\|_p^p + \cdots+\|\bar{x}^{(N)}(T_n)\|_p^p\Big),\\
            \text{therefore, }\;\; \|\hat{x}^{(2)}(T_n)\| \leq\frac{\|\bar{x}^{(2)}(T_n)\|_p}{\Big(C_1^p+1\Big)^\frac{1}{p}\Big(\|\bar{x}^{(2)}(T_n)\|_p^p+\cdots+ \|\bar{x}^{(N)}(T_n)\|_p^p\Big)^{\frac{1}{p}}}\leq \frac{1}{\Big(1+C_1^p\Big)^{\frac{1}{p}}} < \frac{1}{C_1}.
     \end{align*}
The second inequality follows from the fact that \(\|\bar{x}^{(2)}(T_n)\|_p^p \leq \|\bar{x}^{(2)}(T_n)\|_p^p + \cdots+\|\bar{x}^{(2)}(T_n)\|_p^p\). A similar analysis can be carried out to show that $\|\bar{x}^{(j)}(T_n)\|_p^p< \frac{1}{C_1}, 1\leq j \leq N$. Next we show that 
\(\|\hat{x}^{(1)}(T_n)\|_p>\frac{1}{1+\frac{1}{C_1}}.\)
Here we focus on the case when iterates are blowing up. Therefore let $r(n) = \bar{X}^{(1)}(T_n)$. Then,
\begin{align*}
    \|\hat{x}(T_n)\| &= \frac{\|\bar{x}^{(1)}(T_n)\|}{\|\bar{X}^{(1)}(T_n)\|} =\frac{\|\bar{x}^{(1)}(T_n)\|}{\Big(\|\bar{x}^{(1)}(T_n)\|_p^p +\cdots+ \|\bar{x}^{(N)}(T_n)\|_p^p\Big)^{\frac{1}{p}}}\\
    &= \frac{1}{\Big(1+\frac{\|\bar{x}^{(2)}(T_n)\|_p^p+\cdots+\|\bar{x}^{(N)}(T_n)\|_p^p}{\|\bar{x}^{(1)}(T_n)\|_p^p}\Big)^{\frac{1}{p}}}> \frac{1}{\Big(1+\frac{\|\bar{x}^{(2)}(T_n)\|_p^p+\cdots+\|\bar{x}^{(N)}(T_n)\|_p^p}{C_1^p(\|\bar{x}^{(2)}(T_n)\|_p^p+\cdots+\|\bar{x}^{(N)}(T_n)\|_p^p)}\Big)^{\frac{1}{p}}}> \frac{1}{1+\frac{1}{C_1}}.
\end{align*}
Let $\tilde{x}^{(j)}(t-T_n) = x^{(j)}_n(t)$, $\forall 2\leq j \leq N$, $\forall t \in [T_n, T_{n+1}]$. From Lemma \ref{fastest_aux}, $\exists r_{\frac{1}{4}}, c_{\frac{1}{4}}, T_{\frac{1}{4}} >0$ such that 
\[\|x_{c}^{(1),\tilde{x}^{(2)}(t),\ldots,\tilde{x}^{(N)}(t)}(t,\hat{x}^{(1)}(T_n))\|\leq \frac{1}{4}, \forall t \geq T_{\frac{1}{4}}, \forall c \geq c_{\frac{1}{4}},\]
whenever $\tilde{x}^{(j)}(t) \in B^j(0,r_{\frac{1}{4}})$. Choose $C_1 > \max(c_{\frac{1}{4}},\frac{2}{r_{\frac{1}{4}}})$ and $T = T_{\frac{1}{4}}$. Since $\dot{x}^{(j)}(t) = 0, \forall 2 \leq j \leq N$ for the ODE defined in \textbf{(D6)}, $\tilde{x}^{(j)}(t-T_n) = x^{(j)}_n(t) = \hat{x}^{(j)}(T_n), 2 \leq j \leq N$, $\forall t \in [T_{n},T_{n+1}].$ From $\|\hat{x}^{(j)}(T_n)\| < \frac{1}{C_1}, 2\leq j \leq N$, it follows that $\tilde{x}^{(j)}(s) \in B^{j}(0,r_{\frac{1}{4}}), 2 \leq j \leq N$, $\forall s \in [0,T].$ Using Lemma \ref{fast_lemma}(ii), $\|\hat{x}^{(1)}(T_{n+1}^{-}) - x^{(1)}_{n}(T_{n+1})\| < \frac{1}{4}$ for large enough $n$. Also observe that $\|x^{(1)}_{n}(T_{n+1})\| = \|x^{(1),\tilde{x}^{(2)}(t),\ldots,\tilde{x}^{(N)}(t)}_{r(n)}(T_{n+1} - T_{n}, \hat{x}^{(1)}(T_{n}))\| \leq \frac{1}{4}$. Using these, we have $\|\hat{x}^{(1)}(T_{n+1}^{-})\| \leq \|\hat{x}^{(1)}(T_{n+1}^{-})-\bar{x}^{(1)}_{n}(T_{n+1})\| + \|\bar{x}^{(1)}_{n}(T_{n+1})\| \leq \frac{1}{2}$, where $\hat{x}^{(1)}(T_{n+1}^{-}) \stackrel{\triangle}{=} \bar{x}^{(1)}(T_{n+1})/r(n)$. Finally, since \[\frac{\|\bar{x}^{(1)}(T_{n+1})\|}{\|\bar{x}^{(1)}(T_{n})\|} = \frac{\|\hat{x}^{(1)}(T_{n+1}^{-})\|}{\|\hat{x}^{(1)}(T_n)\|},\] 
\begin{align*}
        \text{we have, }\qquad\|\bar{x}^{(1)}(T_{n+1})\| = \frac{\|\hat{x}^{(1)}(T_{n+1}^{-})\|}{\|\hat{x}^{(1)}(T_n)\|} \|\bar{x}^{(1)}(T_n)\|< \frac{\frac{1}{2}}{\frac{1}{1+1/C_1}}\|\bar{x}^{(1)}(T_n)\|.
\end{align*}
Choosing $C_{1} > \max\left(c_{\frac{1}{4}},\frac{2}{r_{\frac{1}{4}}}\right) > 2$, proves the claim.

The claims in (ii) and (iii) follow in a similar manner as Lemma 6(ii)-(iii) of \cite{chandru-SB}. We repeat the arguments for the sake of completeness. 

(ii) We will prove the claim by contradiction. Suppose there exists a monotonically increasing sequence $\{n_{k}\}$ such that $C_{n_{k}}$ $\uparrow \infty$ as $k\rightarrow\infty$ and \(\|\bar{x}^{(1)}(T_n)\| \geq C_{n_{k}}(\|\bar{x}^{(2)}(T_n)\|+\cdots+ \|\bar{x}^{(N)}(T_n)\|)\), on a set of positive probability. From Lemma \ref{fast_lemma} (i), we know that if $\|\bar{x}^{(1)}(T_n)\| > C_{1}(\|\bar{x}^{(2)}(T_n)\|+\cdots+ \|\bar{x}^{(N)}(T_n)\|)$, then $\|\bar{x}^{(1)}(T_k)\|$ falls at an exponential rate into a ball of radius $C_{1}(\|\bar{x}^{(2)}(T_k)\|+\cdots+ \|\bar{x}^{(N)}(T_k)\|)$ for $k\geq n$. Therefore, corresponding to the sequence $\{n_{k}\}$, there exists another sequence $\{n_k'\}$ such that $\|\bar{x}^{(1)}(T_{n_{k-1}'})\| \leq C_1(\|\bar{x}^{(2)}(T_{n_{k}'})\|+\cdots+ \|\bar{x}^{(N)}(T_{n_{k}'})\|)$ for $n_{k-1} \leq n_{k}' \leq n_{k}$, but $\|\bar{x}^{(1)}(T_{n_{k}'})\| > C_{n_{k}}(\|\bar{x}^{(2)}(T_{n_{k}'})\|+\cdots+ \|\bar{x}^{(N)}(T_{n_{k}'})\|)$. From Lemma \ref{fast_lemma} (i), however, we know that the iterates $\bar{x}^{(1)}$ can only grow by a factor of $K^{(1)}$ between $m(n_{k}' - 1)$ and $m(n_{k}')$, leading to a contradiction. Therefore, \(\|\bar{x}^{(1)}(T_n)\| \leq C_{1}^{*}(\|\bar{x}^{(2)}(T_n)\|+\cdots+ \|\bar{x}^{(N)}(T_n)\|)\) a.s. for some $C_{1}^*\geq1$.

(iii) From Lemma \ref{fast_lemma}(ii),  we know that $\forall t \in [T_n,T_{n+1})$, $\|\bar{x}^{(1)}(t)\| \leq K^{(1)}\|\bar{x}^{(1)}(T_n)\|$. Since, $\bar{x}^{(1)}(t)$ is a linear interpolation of the iterates $x^{(1)}_{n}$, therefore, $\|x^{(1)}_{n}\| \leq \sup_{t\in[t(n),t(n+1)]}\|\bar{x}^{(1)}(t)\|$. The claim therefore follows by choosing $K_{1}^{*} = C_{1}^{*}K^{(1)}$. 
\hfill $\square$

Next we consider the intermediate timescales of $\{\alpha^{(l)}_{n}\}_{n\geq 1}, 2\leq l \leq N-1$, and re-define the terms below. Note the abuse of notation here when defining terms such as $T_n$, $r(n)$, $m(n)$, $\hat{M}^{(l)}_{n+1}$ etc., below.
\begin{enumerate}
    \item[\bf{(E1)}] Define
    \[t^{(l)}(n) = \sum_{i=0}^{n-1}\alpha^{(l)}_i,n\geq1 \mbox{ with } t^{(l)}(0) = 0.\]
    Recall that $X_n = (x^{(1)}_n,x^{(2)}_n,\ldots,x^{(N)}_n), n\geq0$.
     For $t \in [t^{(l)}(n), t^{(l)}(n+1)]$, $l=2,\ldots, N$, define
    \[\bar{X}^{(l)}(t) = X_{n} + \left(X_{n+1} - X_{n}\right)\frac{t-t^{(l)}(n)}{t^{(l)}(n+1) - t^{(l)}(n)}.\]
    
    \item[\bf{(E2)}]Given $t^{(l)}(n), n\geq 0$ and a constant $T>0$ define
    \[T_0 = 0,\quad T_n = \min(t^{(l)}(m):t^{(l)}(m)\geq T_{n-1}+T), n\geq1\]
    One can find a subsequence $\{m(n)\}$ such that $T_n = t^{(l)}(m(n))$ $\forall n$, and $m(n) \rightarrow \infty$ as $n\rightarrow \infty$.
    
    \item[\bf{(E3)}] The scaling sequence is defined as:
    \[r(n) = \max(r(n-1)\\|\bar{X}^{(l)}(T_n)\|,1), n\geq1\]
    
    \item[\bf{(E4)}] The scaled iterates for $m(n)\leq k\leq m(n+1)-1$ are defined by:
    \[\hat{x}^{(1)}_{k} = \frac{x_{k}^{(1)}}{r(n)}, \hat{x}^{(2)}_{k} = \frac{x^{(2)}_{k}}{r(n)},\ldots, \hat{x}^{(N)}_{k} = \frac{x_{k}^{(N)}}{r(n)}.\]
Further,    
    \[\hat{x}^{(1)}_{k+1} = \hat{x}^{(1)}_{k} + \alpha^{(1)}_k\left(\frac{h^{(1)}(c\hat{x}^{(1)}_{k},\ldots,c\hat{x}^{(N)}_{k})}{c} + \hat{M}_{k+1}^{(1)}\right)\]
    \[\hat{x}^{(2)}_{k+1} = \hat{x}^{(2)}_{k} + \alpha^{(2)}_k\left(\frac{h^{(2)}(c\hat{x}^{(1)}_{k},\ldots,c\hat{x}^{(N)}_{k})}{c} + \hat{M}_{k+1}^{(2)}\right)\]
    \[\longvdots{2em}\]
    \[\hat{x}^{(l)}_{k+1} = \hat{x}^{(l)}_{k} + \alpha^{(l)}_k\left(\frac{h^{(l)}(c\hat{x}^{(1)}_{k},\ldots,c\hat{x}^{(N)}_{k})}{c} + \hat{M}_{k+1}^{(l)}\right)\]
    \[\hat{x}^{(l+1)}_{k+1} = \hat{x}^{(l+1)}_{k} + \alpha^{(l)}_k\left(\epsilon_{k}^{(l+1)} + \hat{M}_{k+1}^{(l+1)}\right)\]
    \[\longvdots{2em}\]
    \[\hat{x}^{(N)}_{k+1} = \hat{x}^{(N)}_{k} + \alpha^{(l)}_k\left(\epsilon_{k}^{(N)} + \hat{M}_{k+1}^{(N)}\right)\]
    where, $c = r(n)$, and $\forall l\leq j\leq N-1$,
    \[\epsilon_{k}^{(j)} = \frac{\alpha^{(j)}_{k}}{\alpha^{(l)}_{k}}\frac{h^{(j)}(c\hat{x}^{(1)}_k,\ldots,c\hat{x}^{(N)}_k)}{c},\]
    \[\hat{M}_{k+1}^{(j)} =  \frac{M_{k+1}^{(j)}}{r(n)},\mbox{ \quad for } 1\leq j\leq l, \qquad \hat{M}_{k+1}^{(j)} = \frac{\alpha^{(j)}_k}{\alpha^{(l)}_k}\frac{M_{k+1}^{(j)}}{r(n)},\mbox{ \quad for } l+1\leq j\leq N-1.\]
    
    \item[\bf{(E5)}] Next, we define the linearly interpolated trajectory for the scaled iterates $\hat{X}^{(l)}_n=(x^{(1)}_{n},x^{(2)}_{n},\ldots,x^{(N)}_{n})$ for $t \in [t^{(l)}(n), t^{(l)}(n+1)]$ as follows:
    \[\hat{X}^{(l)}(t) = \hat{X}^{(l)}_n + (\hat{X}^{(l)}_{n+1} - \hat{X}^{(l)}_n)\frac{t - t(n)}{t(n+1) - t(n)}.\]
    
    \item[\bf{(E6)}] Let \(X^{(l)}_{n}(t) = (x^{(1)}_n(t),\ldots,x^{(N)}_n(t)), t\in [T_n,T_{n+1}]\) denote the trajectory of the ODE:
    \[\dot{x}^{(1)}(t) = h^{(1)}_{r(n)}(x^{(1)}(t),\ldots,x^{(N)}(t)),\]
    \[\dot{x}^{(2)}(t) = h^{(2)}_{r(n)}(x^{(2)}(t),\ldots,x^{(N)}(t)),\]
    \[\longvdots{2em}\]
    \[\dot{x}^{(l)}(t) = h^{(l)}_{r(n)}(x^{(l)}(t),\ldots,x^{(N)}(t)),\]
    \[\dot{x}^{(l+1)}(t) = 0\]
    \[\longvdots{2em}\]
    \[\dot{x}^{(N)}(t) = 0\]
    with $x^{(j)}_n(T_n) = \hat{x}^{(j)}(T_n)$, $\forall 1 \leq j\leq N$. We refer the reader to Assumption {\bf (B.N.i)} for the definition of $h^{(i)}_{r(n)}(x^{(i)}(t),\ldots,x^{(N)}(t))$, $i=1,\ldots,N$.
\end{enumerate}  
\begin{lemma}
    \label{intermediate_aux}
        Let $x^{(l)} \in B^{l}(0,1) \subset \mathbb{R}^{d_l},x^{(j)} \in W^{(j)}\subset\mathbb{R}^{d_j}, l+1\leq j \leq N$ and let \textbf{(B.N.l)} hold. Then given $\epsilon>0, \exists c_{\epsilon}\geq 1, r_{\epsilon}>0$ and $T_{\epsilon}>0$ such that for any external inputs satisfying \(\tilde{x}^{(j)}(s) \in B^{j}(x^{(j)},r_{\epsilon}), l+1\leq j \leq N-1\), $\forall s \in [0,T], \forall c > c_{\epsilon},t\geq T_{\epsilon}$, 
        \[\|x_c^{(l),\tilde{x}^{(l+1)}(t),\ldots,\tilde{x}^{(N)}(s)}(t,x^{(l)}) - \lambda_{\infty}^{(l)}(x^{(l+1)},\ldots,x^{(N)})\| \leq 2\epsilon,\]
with ${\displaystyle x_c^{(l),\tilde{x}^{(l+1)}(t),\ldots,\tilde{x}^{(N)}(s)}(t,x^{(l)})}$ defined analogously as in Lemma~\ref{fastest_aux}.
\end{lemma}

\begin{lemma}
    \label{intermediate_lemma}
    Under \textbf{(A:1)}-\textbf{(A:3)},
    \begin{enumerate}
        \item[(i)] For \(0\leq k \leq m(n+1)-m(n)\), and $2 \leq l \leq N-1$
        \(\|\hat{X}^{(l)}(t(m(n)+k))\| \leq K^{(l)}\)  a.s. for some constant $K^{(l)}>0.$
        
        \item[(ii)] We have,
            \(\lim_{n\rightarrow\infty} \sup_{t\in [T_n,T_{n+1}]}
            \|\hat{X}^{(l)}(t) - X^{(l)}_n(t)\| = 0 \mbox{ a.s. } \)

    \end{enumerate}
\end{lemma}
The proof of the above two Lemmas follows in a similar manner as the proof of Lemmas 5 and 9, respectively, of \cite{chandru-SB}.
\begin{lemma}
    \label{intermediate_main}
    Assume \textbf{(A:1)}-\textbf{(A:3)}, \textbf{(B.N.i)}$_{1\leq i \leq N}$ and \textbf{(C.N.i)}$_{1\leq i \leq N}$ hold. Then,
    \begin{enumerate}
        \item[(i)] For $n$ large enough, there exists $T>0$, such that if $\|\bar{x}^{(l)}(T_n)\| > C_{l}(\|\bar{x}^{(l+1)}(T_n)\|+\cdots+\|\bar{x}^{(N)}(T_n)\|)$, for some $C_l>0$, then $\|\bar{x}^{(l)}(T_{n+1})\| < \frac{3}{4}\|\bar{x}^{(l)}(T_n)\|$.
        \item[(ii)] $\|\bar{x}^{(l)}(T_n)\| \leq C_{l}^*\left(\|\bar{x}^{(l+1)}(T_n)\|+\cdots+\|\bar{x}^{(N)}(T_n)\|\right)$, for some $C_l^*>1$.
        \item[(iii)] $\|x^{(l)}_n\| \leq K_l^*( \|x^{(l+1)}_n\|+\cdots+\|x^{(N)}_n\|)$, for some $K_l^*>0$.
    \end{enumerate}
\end{lemma}
\emph{Proof.} 
We use an inductive argument to show that the lemma holds. 
The base case $l=1$ holds from Lemma \ref{fast_main}.    
Assume the claim holds for $l-1$. We now show that it holds for $l$. 
As in Lemma \ref{fast_main}, we can show $\|\bar{x}^{(j)}(T_n)\|_p^p< \frac{1}{C_l}, l+1\leq j \leq N$.
Next we show that $\|\hat{x}^{(l)}(T_n)\|>\frac{1}{((l-1)(C_{l,max}^{*})^{l-1}+1)(1+\frac{1}{C_l})}$, where 
\(C_{l,max}^{*} = \max_{1\leq j\leq l-1}(C_j^{*}).\)
Here again we are considering the case when the iterates are blowing up. Therefore let $r(n) = \|\bar{X}^{(l)}(T_n)\|$. Now, from the inductive step, we know that $\|\bar{x}^{(j)}(T_n)\|\leq C_j^*(\|\bar{x}^{(j+1)}(T_n)\| +\cdots+ \|\bar{x}^{(N)}(T_n)\|),l \leq j \leq N$ and therefore,
\begin{gather*}
    r(n) \leq \sum_{j=1}^{l-1} C_j^*( \|\bar{x}^{(j+1)}(T_n)\| +\cdots+ \|\bar{x}^{(N)}(T_n)\|)
    + \|\bar{x}^{(l)}(T_n)\| +\cdots+ \|\bar{x}^{(N)}(T_n)\|\\
    = \Big(\sum_{j=1}^{l-1}\prod_{v=1}^{l-j}C_{v}^*+1\Big)\big(\|\bar{x}^{(l)}(T_n)\| +\cdots+ \|\bar{x}^{(N)}(T_n)\|\big)
    \leq \Big((l-1)(C_{l,max}^{*})^{(l-1)}+1\Big)\Big(\|\bar{x}^{(l)}(T_n)\| +\cdots+ \|\bar{x}^{(N)}(T_n)\|\Big).
\end{gather*}
We thus have,
\begin{gather*}
        \|\hat{x}^{(l)}(T_n)\|_p\geq \frac{\|\bar{x}^{(l)}(T_n)\|}{((l-1)(C_{l,max}^{*})^{l-1}+1)(\|\bar{x}^{(l)}(T_n)\| +\cdots+ \|\bar{x}^{(N)}(T_n)\|)}\\
        > \frac{1}{((l-1)(C_{l,max}^{*})^{l-1}+1)(1+\frac{1}{C_l})}
        = \frac{\epsilon}{1+\frac{1}{C_l}},
\end{gather*}
where ${\displaystyle \epsilon = \frac{1}{(l-1)(C_{l,max}^{*})^{l-1}+1}}$.
Let $\tilde{x}^{(j)}(t-T_n) = x^{(j)}_n(t)$, $\forall 2\leq j \leq N$, $\forall t \in [T_n, T_{n+1}]$. From Lemma \ref{intermediate_aux}, $\exists r_{\epsilon/4}$, $c_{\epsilon/4}$, $T_{\epsilon/4}>0$ such that 
\begin{gather*}
    \|x_{c}^{(l),\tilde{x}^{(l+1)}(t),\ldots,\tilde{x}^{(N)}(t)}(t,\hat{x}^{(1)}(T_n))\| \leq \epsilon,
    \end{gather*}
$\forall t \geq T_{\epsilon/4}$, $\forall c \geq c_{\epsilon/4}$,
whenever $\tilde{x}^{(j)}(t) \in B^j(0,r_{\epsilon/4})$, $l+1 \leq j \leq N$. Choose $T = T_{\epsilon/4}$. Since $\dot{x}^{(j)}(t) = 0$, $l+1\leq j \leq N$ for the ODE defined in $\textbf{(E6)}$ and $\tilde{x}^{(j)}(t-T_n) = x^{(j)}_n(t) = \hat{x}^{(j)}(T_n)$ $\forall t \in [T_{n},T_{n+1}]$ and we choose $C_{l} > \max\left(c_{\epsilon/4},\frac{2}{r_{\epsilon/4}}\right)$, from $\|\hat{x}^{(j)}(T_n)\| < \frac{1}{C_l}$, it follows that $\tilde{x}^{(j)}(s) \in B^j(0,r_{\epsilon/4})$, $\forall s \in [0,T].$ Using Lemma \ref{intermediate_lemma}(ii), $\exists \Gamma_1>0$ s.t. $\|\hat{x}^{(l)}(T_{n+1}^{-}) - x^{(l)}_{n}(T_{n+1})\| < \frac{\epsilon}{4}$ for large enough $n$ and $r(n)>\Gamma_l$. Choose $C_l >$
$\max\left(c_{\epsilon/4},\frac{2}{r_{\epsilon/4}}, C_1\right)$.
Also observe that 
\begin{gather*}
    \|x^{(l)}_{n}(T_{n+1})\| = \|x_{r(n)}^{(l),\tilde{x}^{(l+1)}(t),\ldots,\tilde{x}^{(N)}(t)}(T_{n+1} - T_{n}, \hat{x}^{(l)}(T_{n}))\|
    \leq \frac{\epsilon}{4}.
\end{gather*}
Using these, we have
\begin{gather*}
    \|\hat{x}^{(l)}(T_{n+1}^{-})\| \leq \|\hat{x}^{(l)}(T_{n+1}^{-})-x^{(l)}_{n}(T_{n+1})\| + \|x^{(l)}_{n}(T_{n+1})\| \leq \frac{\epsilon}{2}.
\end{gather*}
Finally, since \[\frac{\|\bar{x}^{(l)}(T_{n+1})\|}{\|\bar{x}^{(l)}(T_{n})\|} = \frac{\|\hat{x}^{(l)}(T_{n+1}^{-})\|}{\|\hat{x}^{(l)}(T_n)\|},\] we have 
\begin{align*}
        \|\bar{x}^{(l)}(T_{n+1})\| = \frac{\|\hat{x}^{(l)}(T_{n+1}^{-})\|}{\|\hat{x}^{(l)}(T_n)\|} \|\bar{x}^{(l)}(T_n)\|< \frac{\frac{\epsilon}{2}}{\frac{\epsilon}{((l-1)((1+1/C_l)}}\|\bar{x}(T_n)\|<\frac{1+\frac{1}{C_l}}{2}
\end{align*}
Choosing $C_l>2$ proves the claim.

(ii) As before, we will prove the claim by contradiction. Suppose there exists a monotonically increasing sequence $\{n_{k}\}$ such that $C_{n_{k}}$ $\uparrow \infty$ as $k\rightarrow\infty$ and \(\|\bar{x}^{(l)}(T_n)\| \geq C_{n_{k}}(\|\bar{x}^{(l+1)}(T_n)\|+\cdots+ \|\bar{x}^{(N)}(T_n)\|)\), on a set of positive probability. From Lemma \ref{intermediate_lemma} (i), we know that if $\|\bar{x}^{(l)}(T_n)\| > C_{l}(\|\bar{x}^{(l+1)}(T_n)\|+\cdots+ \|\bar{x}^{(N)}(T_n)\|)$, then $\|\bar{x}^{(l)}(T_k)\|$ falls at an exponential rate into a ball of radius $C_{l}(\|\bar{x}^{(l+1)}(T_k)\|+\cdots+ \|\bar{x}^{(N)}(T_k)\|)$ for $k\geq n$. Therefore, corresponding to the sequence $\{n_{k}\}$, there exists another sequence $\{n_k'\}$ such that $\|\bar{x}^{(l)}(T_{n_{k-1}'})\| \leq C_l(\|\bar{x}^{(l+1)}(T_{n_{k}'})\|+\cdots+ \|\bar{x}^{(N)}(T_{n_{k}'})\|)$ for $n_{k-1} \leq n_{k}' \leq n_{k}$, but $\|\bar{x}^{(l)}(T_{n_{k}'})\| > C_{n_{k}}(\|\bar{x}^{(l+1)}(T_{n_{k}'})\|+\cdots+ \|\bar{x}^{(N)}(T_{n_{k}'})\|)$. From Lemma \ref{intermediate_lemma} (i), however, we know that the iterates $\bar{x}^{(l)}$ can only grow by a factor of $K^{(l)}$ between $m(n_{k}' - 1)$ and $m(n_{k}')$, leading to a contradiction. Therefore, \(\|\bar{x}^{(l)}(T_n)\| \leq C_{l}^{*}(\|\bar{x}^{(l+1)}(T_n)\|+\cdots+ \|\bar{x}^{(N)}(T_n)\|)\) a.s. for some $C_{l}^*\geq1$.

(iii) From Lemma \ref{fast_lemma}(ii),  we know that $\forall t \in [T_n,T_{n+1})$, $\|\bar{x}^{(l)}(t)\|\leq K^{(l)}\|\bar{x}^{(l)}(T_n)\|$. Since, $\bar{x}^{(l)}(t)$ is a linear interpolation of the iterates $x^{(l)}_{n}$, therefore, $\|x^{(l)}_{n}\| \leq \sup_{t\in[t(n),t(n+1)]}\|\bar{x}^{(l)}(t)\|$. The claim therefore follows by choosing $K_{l}^{*} = C_{l}^{*}K^{(l)}$. 
\hfill $\square$

Finally, we consider the slowest timescale recursion corresponding to the step-size $\{\alpha^{(N)}_{n}\}_{n\geq 0}$, and re-define the terms used. Note (again) the abuse of notation in definitions of terms such as $T_n$, $m(n)$, $r(n)$, $\hat{M}^j_{n+1}$, etc., below.
\begin{enumerate}
    \item[\bf{(F1)}] Define
    \({\displaystyle t^{(N)}(n) = \sum_{i=0}^{n-1}\alpha^{(N)}_i}\), \(n\geq1\) with  $t(0) = 0$.
    Let $X^{(N)}_k = (x^{(1)}_k,x^{(2)}_k,\ldots,x^{(N)}_k), k\geq0$, and for $t \in [t^{(N)}(n), t^{(N)}(n+1)]$,
    \[\bar{X}^{(N)}(t) = X^{(N)}_{n} + \left(X^{(N)}_{n+1} - X^{(N)}_{n}\right)\frac{t-t^{(N)}(n)}{t^{(N)}(n+1) - t^{(N)}(n)}.\]
    
    \item[\bf{(F2)}]Given $t^{(N)}(n), n\geq 0$ and constant $T>0$, define
    \[T_n = \min(t^{(N)}(m):t^{(N)}(m)\geq T_{n-1}+T), n\geq1,\]
  with $T_0=0$. One can find a subsequence $\{m(n)\}$ such that $T_n = t^{(N)}(m(n))$, $\forall n$, and $m(n) \rightarrow \infty$ as $n\rightarrow \infty$. 
    
    \item[\bf{(F3)}] The scaling sequence is defined as
    \(r(n) = \max(r(n-1)\), \(\|\bar{X}^{(N)}(T_n)\|,1)\), \(n\geq1\).
    
    \item[\bf{(F4)}] The scaled iterates for $m(n)\leq k\leq m(n+1)-1$ are given by
    \[\hat{x}^{(1)}_{k} = \frac{x_{k}^{(1)}}{r(n)}, \hat{x}^{(2)}_{k} = \frac{x^{(2)}_{k}}{r(n)},\ldots, \hat{x}^{(N)}_{k} = \frac{x_{k}^{(N)}}{r(n)}.\]
    The corresponding updates are as follows:
    \[\hat{x}^{(1)}_{k+1} = \hat{x}^{(1)}_{k} + \alpha^{(1)}_k\left(\frac{h^{(1)}(c\hat{x}^{(1)}_{k},\ldots,c\hat{x}^{(N)}_{k})}{c} + \hat{M}_{k+1}^{(1)}\right),\]
    \[\longvdots{2em}\]
    \[\hat{x}^{(N)}_{k+1} = \hat{x}^{(N)}_{k} + \alpha^{(N)}_k\left(\frac{h^{(N)}(c\hat{x}^{(1)}_{k},\ldots,c\hat{x}^{(N)}_{k})}{c} + \hat{M}_{k+1}^{(N)}\right),\]
    where, $c = r(n)$, and $\forall 1\leq j\leq N-1$,
    \({\displaystyle \hat{M}_{k+1}^{(j)} =  \frac{M_{k+1}^{(j)}}{r(n)}}\), for $1\leq j\leq N$.
    
    \item[\bf{(F5)}] Next, we define the linearly interpolated trajectory for the scaled iterates $\hat{X}^{(N)}_n=(x^{(1)}_{n},x^{(2)}_{n},\ldots,x^{(N)}_{n})$ as follows: For $t \in [t^{(N)}(n), t^{(N)}(n+1)]$,
    \[\hat{X}^{(N)}(t) = \hat{X}^{(N)}_n + (\hat{X}^{(N)}_{n+1} - \hat{X}^{(N)}_n)\frac{t - t^N(n)}{t^N(n+1) - t^N(n)}.\]
    
    \item[\bf{(F6)}] Let \(X_{n}^{(N)}(t) = (x^{(1)}_n(t),\ldots,x^{(N)}_n(t)), t\in [T_n,T_{n+1}]\), denote the trajectory of the ODE:
    \[\dot{x}^{(1)}(t) = h^{(1)}_{r(n)}(x^{(1)}(t),\ldots,x^{(N)}(t)),\]
    \[\dot{x}^{(2)}(t) = h^{(2)}_{r(n)}(x^{(2)}(t),\ldots,x^{(N)}(t)),\]
    \[\longvdots{2em}\]
    \[\dot{x}^{(N)}(t) = h^{(N)}_{r(n)}(x^{(N)}(t)),\]
    with $x^{(j)}_n(T_n) = \hat{x}^{(j)}(T_n)$, $1 \leq j\leq N.$
\end{enumerate}  

\begin{lemma}
    \label{slow_aux}
        Let $x^{(N)} \in B^{N}(0,1) \subset \mathbb{R}^{d_N}$ and let Assumption $\bf{(B.N.N)}$ hold. Then given $\epsilon>0, \exists c_{\epsilon}\geq 1, r_{\epsilon}>0$ and $T_{\epsilon}>0$, then 
        \[\|x^{(N)}_c(t,x^{(N)})\| \leq 2\epsilon, \mbox{\quad} \forall c > c_{\epsilon}.\]
    \end{lemma}
    
    \begin{lemma}
    \label{slow_lemma}
    Under \textbf{(A:1)}-\textbf{(A:3)},
    \begin{enumerate}
        \item[(i)] For \(0\leq k \leq m(n+1)-m(n)\),
        \(\|\hat{X}^{(N)}(t(m(n)+k))\| \leq K^{(N)}\)  a.s. for some constant $K^{(N)}>0.$
        \item[(ii)] For sufficiently large $n$, we have  \(\sup_{[T_n,T_{n+1})}\|\hat{x}^{(N)}(t) - x^{(N)}_n(t)\| \rightarrow 0\) almost surely as \(c\rightarrow\infty.\)
    \end{enumerate}
\end{lemma}
As before, the proof of the above two Lemmas follows in a similar manner as the proof of Lemmas 5 and 9, respectively, of \cite{chandru-SB}.
\begin{lemma}
\label{slow_main}
    Under assumptions \textbf{(A:1)}-\textbf{(A:3)}, \textbf{(B.N.i)}$_{1\leq i \leq N}$ and \textbf{(C.N.i)}$_{1\leq i \leq N}$, we have:
    \begin{itemize}
        \item[(i)] For $n$ large, $\exists T$ such that if $\|\bar{x}^{(N)}(T_n)\|>C_{N}$, for some $C_{N}>0$ then $\|\bar{x}^{(N)}(T_{n+1})\|<\frac{1}{2}\|\bar{x}^{(N)}(T_n)\|$.
        \item[(ii)] $\|\bar{x}^{(N)}(T_n)\| \leq C_N^*$ for some $C_N^*>0$.
        \item[(iii)] $\sup_{n}\|x^{(N)}_n\|<\infty$ a.s.
    \end{itemize}
\end{lemma}

\emph{Proof.} 
(i) From Lemmas \ref{fast_main} and \ref{intermediate_main}, we know that 
    \begin{gather*}
        r(n) \leq \sum_{j=1}^{N-1} C_j^*(\|\bar{x}^{(j+1)}(T_n)\| +\cdots+ \|\bar{x}^{(N)}(T_n)\|)
        = \Big(\sum_{j=1}^{N-1}\prod_{v=1}^{N-j}C_{v}^*+1\Big)\Big(\|\bar{x}^{(N)}(T_n)\|\Big).
    \end{gather*}
    
    \begin{gather*}
        \text{Therefore,}\; \|\hat{x}^{(N)}(T_n)\| = \frac{\|\bar{x}^{(N)}(T_n)\|}{r(n)}
            > \frac{1}{\sum_{j=1}^{N-1}\Big(\prod_{v=1}^{N-j}C_{v}^*+1\Big)}
            > \frac{1}{((N-1)(C_{N,max}^{*})^{N-1}+1)} \stackrel{\triangle}{=}\tilde{\epsilon},
    \end{gather*}
    where, 
    \({\displaystyle C_{N,max}^{*} = \max_{1\leq j\leq N-1}(C_j^{*})}.\)
     Since, $0\in\mathbb{R}^{d_3}$ is the unique globally asymptotically stable equilibrium, using Lemma \ref{slow_aux}, $\exists c_{\tilde{\epsilon}/4},  T_{\tilde{\epsilon}/4} >0$, such that
     \(\|x^{(N)}_{c}(t,x^{(N)})\|<\frac{\tilde{\epsilon}}{4},\)
     $\forall c \geq c_{\tilde{\epsilon}/4}, t \geq T_{\tilde{\epsilon}/4}$. 
     Also, for sufficiently large $n$, from Lemma \ref{slow_lemma}(ii), $\exists \Gamma_2>0$ such that 
    $\|\hat{x}^{(N)}(T_{n+1}^{-})-x^{(N)}_{n}(T_{n+1})\| < \frac{\tilde{\epsilon}}{4}$ for $r(n)>\Gamma_2$. 
    We pick 
    \(C_{N} = \max(c_{\tilde{\epsilon}/4},\Gamma_2)\)
    and $T= T_{\tilde{\epsilon}/4}$. 
    For $n$ large, it follows that 
    \begin{gather*}
        \|\hat{x}^{(N)}(T_{n+1}^{-})\| \leq \|\hat{x}^{(N)}(T_{n+1}^{-})-x^{(N}_{n}(T_{n+1})\| + \|x^{(N)}_{n}(T_{n+1})\| \leq \frac{\tilde{\epsilon}}{4}.
    \end{gather*}
    Finally, since \({\displaystyle \frac{\|\bar{x}^{(N)}(T_{n+1})\|}{\|\bar{x}^{(N)}(T_{n})\|} = \frac{\|\hat{x}^{(N)}(T_{n+1}^{-})\|}{\|\hat{x}^{(N)}(T_n)\|}}\), it follows that
    \({\displaystyle \|\bar{x}^{(N)}(T_{n+1})\| < \frac{1}{2}\|\bar{x}^{(N)}(T_{n})\|}\).

The claims in (ii) and (iii) now follow in a similar manner as Theorem 10 (iii)-(iv), respectively, of \cite{chandru-SB}, and the arguments are repeated here for completeness.

(ii) We will prove the claim by contradiction. Suppose there exists a monotonically increasing sequence $\{n_{k}\}$ such that $C_{n_{k}}$ $\uparrow \infty$ as $k\rightarrow\infty$ and \(\|\bar{x}^{(N)}(T_n)\| \geq C_{n_{k}}\), on a set of positive probability. From Lemma \ref{slow_lemma} (i), we know that if $\|\bar{x}^{(N)}(T_n)\| > C_{N}$, then $\|\bar{x}^{(N)}(T_k)\|$ falls at an exponential rate into a ball of radius $C_{N}$ for $k\geq n$. Therefore, corresponding to the sequence $\{n_{k}\}$, there exists another sequence $\{n_k'\}$ such that $\|\bar{x}^{(N)}(T_{n_{k-1}'})\| \leq C_N$ for $n_{k-1} \leq n_{k}' \leq n_{k}$, but $\|\bar{x}^{(N)}(T_{n_{k}'})\| > C_{n_{k}}$. From Lemma \ref{slow_lemma} (i), however, we know that the iterates $\bar{x}^{(N)}$ can only grow by a factor of $K^{(N)}$ between $m(n_{k}' - 1)$ and $m(n_{k}')$, leading to a contradiction. Therefore, \(\|\bar{x}^{(N)}(T_n)\| \leq C_{N}^{*}\) a.s. for some $C_{N}^*\geq1$.

(ii) From the previous part we have $\|\bar{x}^{(N)}\| \leq C_{N}^{*}$, and from Lemma \ref{slow_lemma} (i), we have $\|\bar{x}^{(N)}(t)\| \leq K^{(N)}\|\bar{x}^{(N)}(T_n)\|$, $\forall t \in [T_n, T_{n+1})$. Therefore, $\|x_{n}\| \leq K^{(N)}C_{N}^{*}$ almost surely. 
\hfill $\square$

The fact that $\sup_{n}{\|x_{n}^{(l)}\|} < \infty$ a.s. $\forall l \in \{1,\ldots,N\}$ follows by combining Lemma \ref{slow_main} (iii) and Lemma \ref{intermediate_main} (iii).

\section{Application to Reinforcement Learning Algorithms}
\label{sec_app}
In this section we use our results on $N$-timescale recursions to show stability and convergence of two Reinforcement Learning (RL) algorithms: Gradient Temporal Difference (GTD) with momentum for policy evaluation, and Constrained Actor Critic for policy optimization in constrained RL.


In the standard RL framework, an agent interacts with a stochastic and dynamic environment. At each discrete time step $t$, the agent is in state $s_t \in \mathcal{S}$, picks an action $a_t \in \mathcal{A}$, receives a reward $r_{t+1} \equiv r(s_{t},a_t,s_{t+1}) \in \mathcal{R}$ and probabilistically moves to another state $s_{t+1} \in \mathcal{S}$. The tuple $(\mathcal{S},\mathcal{A},\mathbb{P},\mathcal{R},\gamma)$ constitutes a Markov Decision Process (MDP). Here $\mathcal{S}$ and $\mathcal{A}$ are assumed finite. Also, $\gamma\in (0,1)$ is the discount factor. A policy $\pi:\mathcal{S}\times\mathcal{A}\rightarrow[0,1]$ is a mapping that defines the probability of picking an action in a state. We let $P^{\pi}(s'|s)$ denote the probability of transition to state $s'$ from state $s$ when an action is chosen as per $\pi$. We let $\{d^{\pi}(s)\}_{s \in \mathcal{S}}$ denote the steady-state distribution for the Markov chain induced by $\pi$. The matrix $D$ is a $n \times n$ diagonal matrix with elements $d^{\pi}(s)$ on its diagonals with $n$ being the number of states. The value function corresponding to state $s$ under policy $\pi$ for state $s$ is defined by:
\[V^{\pi}(s) = \mathbb{E}_{\pi}\left[ \sum_{t=0}^{\infty} \gamma^{t} R_{t+1}|s_{0} = s\right].\]

\subsection{Gradient Temporal Difference with Momentum}
With linear function approximation for policy evaluation (i.e., for a fixed $\pi$), the goal is to estimate $V^{\pi}(s)$ from samples of the form $(s_t,r_{t+1},s_{t+1})$ through a linear model $V_{\theta}(s) = \theta^{T}\phi(s)$. Here $\phi(s) \equiv \phi_{s}$ is a feature vector associated with the state $s$ and $\theta$ is the associated parameter vector. The TD-error is defined by \(\delta_{t} = r_{t+1} + \gamma \theta_{t}^{T}\phi_{t+1} - \theta_{t}^{T}\phi_{t}\). The feature matrix $\Phi$ is an $n\times d$ matrix where the $s^{th}$ row is $\phi(s)^T$. In the following, we consider the i.i.d setting, where the tuple $(\phi_{t},\phi_{t}'$) (with $\phi_{t}' \equiv \phi_{t+1}$ ) is drawn independently from the stationary distribution $\{d_\pi(s)\}$. Let \(\bar{A} \stackrel{\triangle}{=} \mathbb{E}[\phi_{t}(\gamma\phi_{t}'-\phi_{t})^{T}]\) and \(\bar{b} \stackrel{\triangle}{=} \mathbb{E}[r_{t+1}\phi_t]\), where the expectations are w.r.t. the stationary distribution of the induced Markov chain. The matrix $\bar{A}$ is known to be negative definite (see \cite{Maei_PhD,TsitsiklisVanRoy}). In the off-policy setting, the behaviour policy $\pi$ is used to sample trajectories from the MDP while the target policy $\mu$ is the one whose associated value function needs to be approximated. Let $\rho_{t} = \frac{\pi(a_{t}|s_{t})}{\mu(a_{t}|s_{t})}$ denote the importance sampling ratio. Gradient TD algorithms are a class of TD algorithms that are convergent even in the off-policy setting. 
We first present the iterates associated with the algorithms GTD2 and TDC, see \cite{FastGradient}.
\begin{enumerate}
    \item \textbf{GTD2}: 
    \begin{gather}
        \label{gtd2_1}
        \theta_{t+1} = \theta_{t} + \alpha_{t}(\phi_{t} - \gamma\phi_{t}')\phi_{t}^{T}u_{t},\\
        \label{gtd2_2}
        u_{t+1} = u_{t} + \beta_{t}(\delta_{t} - \phi_{t}^{T}u_{t})\phi_{t}.
    \end{gather}
    \item \textbf{TDC}: 
    \begin{gather}
        \label{tdc_1}
        \theta_{t+1} = \theta_{t} + \alpha_{t}\delta_{t}\phi_{t} - \alpha_{t}\gamma \phi_{t}'(\phi_{t}^{T}u_{t}),\\
        \label{tdc_2}
        u_{t+1} = u_{t} + \beta_{t}(\delta_{t} - \phi_{t}^{T}u_{t})\phi_{t}.
    \end{gather}
\end{enumerate}
In the GTD algorithm, the objective function considered is the Norm of Expected Error defined as $NEU(\theta) = \mathbb{E}[\delta\phi]$ and the algorithm is derived by expressing the gradient direction as $-\frac{1}{2} \nabla NEU(\theta)$ = $\mathbb{E}\left[ (\phi-\gamma \phi')\phi^T\right] \mathbb{E}[\delta(\theta)\phi]$. Here $\phi' \equiv \phi(s')$. Since the expectation becomes biased by the correlation of the two terms if both the terms are sampled separately, an estimate of the second expectation is maintained as a long-term quasi-stationary estimate while samples for the first expectation are used. For GTD2 and TDC, a similar approach is used on the objective function Mean Square Projected Bellman Error defined as $MSPBE(\theta) = \|V_{\theta} - \Pi T^{\pi} V_{\theta}\|_{D}$, where for any $x\in \mathbb{R}^n$,
$\|x\|_D = \sqrt{x^TDx}$. Here, $\Pi$ is the projection operator that projects vectors to the subspace $\{\Phi\theta|\theta\in\mathbb{R}^{d}\}$ and $T^{\pi}$ is the Bellman operator defined as $T^{\pi}V = R^{\pi} + \gamma P^{\pi}V$. It was shown in all the three cases that $\theta_{n}\rightarrow\theta^{*} = -\bar{A}^{-1}\bar{b}$.
\subsubsection{Three Timescale Gradient TD Algorithms with Momentum}
\hfill \break
We consider the Gradient TD algorithms with an added heavy ball term to the first iterate. 
\begin{enumerate}
    \item GTD2 with momentum (\textbf{GTD2-M-3TS}): 
    \begin{gather}
    \label{gtd2_M_1}
        \theta_{t+1} = \theta_{t} + \alpha_{t}(\phi_{t} - \gamma\phi_{t}')\phi_{t}^{T}u_{t} + \eta_{t}(\theta_{t} - \theta_{t-1}),\\
    \label{gtd2_M_2}
         u_{t+1} = u_{t} + \beta_{t}(\delta_{t} - \phi_{t}^{T}u_{t})\phi_{t}.
    \end{gather}

    \item TDC with momentum (\textbf{TDC-M-3TS}):
    \begin{gather}
    \label{tdc_M_1}
        \theta_{t+1} = \theta_{t} + \alpha_{t}(\delta_{t}\phi_{t} - \gamma\phi_{t}'(\phi_{t}^{T}u_{t}))+ \eta_{t}(\theta_{t} - \theta_{t-1}),\\
    \label{tdc_M_2}
         u_{t+1} = u_{t} + \beta_{t}(\delta_{t} - \phi_{t}^{T}u_{t})\phi_{t}.
    \end{gather}
\end{enumerate}
The momentum parameter $\eta_{t}$ is chosen as in \cite{webpage} as 
\(\eta_t =\frac{\varrho_t-w\alpha_t}{\varrho_{t-1}},\) where $\{\varrho_t\}$ is a positive sequence and $w\in\mathbb{R}$ is a constant. 
We let 
\(\frac{\theta_{t+1} - \theta_{t}}{\varrho_{t}} = v_{t+1}, \xi_{t} = \frac{\alpha_{t}}{\varrho_{t}} \mbox{ and } \varepsilon_{t} = v_{t+1} - v_{t}.\) Then
the iterates for \textbf{GTD2-M-3TS} can then be decomposed into the three recursions as below:
\begin{gather}
    \label{GTD2-M-1_app}
    v_{t+1} = v_{t} + \xi_{t}\left((\phi_{t} - \gamma\phi_{t}')\phi_{t}^{T}u_{t} - w v_{t}\right),\\
    \label{GTD2-M-2_app}
    u_{t+1} = u_{t} + \beta_{t} (\delta_{t}\phi_{t} - \phi_t\phi_t^Tu_{t}),\\
    \label{GTD2-M-3_app}
    \theta_{t+1} = \theta_{t} + \varrho_{t}(v_{t} + \varepsilon_{t}).
\end{gather}
Similarly, the iterates for \textbf{TDC-M-3TS} can be decomposed as:
\begin{gather}
    \label{TDC-M-1_app}
    v_{t+1} = v_{t} + \xi_{t}\left(\delta_t\phi_t - \gamma\phi_{t}'\phi_{t}^{T}u_{t} - w v_{t}\right),\\
    \label{TDC-M-2_app}
    u_{t+1} = u_{t} + \beta_{t} (\delta_{t}\phi_{t} - \phi_t\phi_t^Tu_{t}),\\
    \label{TDC-M-3_app}
    \theta_{t+1} = \theta_{t} + \varrho_{t}(v_{t} + \varepsilon_{t}).
\end{gather}
Consider the following assumptions:
\begin{assumption}
\label{3A1}
All rewards $r(s,s')$ and features $\phi(s)$ are bounded, i.e., $r(s,s')\leq 1$ and $\|\phi(s)\|\leq 1$ $\forall s,s' \in \mathcal{S}$. Also, the matrix $\Phi$ has full rank, where $\Phi$ is an $n\times d$ matrix where the s$^{th}$ row is $\phi(s)^T$.
\end{assumption}
\begin{assumption}
\label{3A2}
The step-sizes satisfy \(\xi_{t}>0,\beta_{t}>0, \varrho_{t}>0 \mbox{ }\forall t\), 
\[ \sum_{t}\xi_t = \sum_{t}\beta_t = \sum_{t}\varrho_t = \infty,\mbox{ }
    \sum_{t}(\xi_t^2+\beta_t ^2+\varrho_t^2) < \infty,\]
    \({\displaystyle \frac{\beta_t}{\xi_t}\rightarrow 0, \frac{\varrho_t}{\beta_t}\rightarrow 0 \mbox{ as } t \rightarrow \infty}\),
and the momentum parameter satisfies:
 \({\displaystyle \eta_{t} = \frac{\varrho_{t}-w\alpha_{t}}{\varrho_{t-1}}}\).
\end{assumption}
\begin{assumption}
\label{3A3}
The samples ($\phi_{t},\phi_{t}'$) are drawn i.i.d from the stationary distribution of the Markov chain induced by the target policy $\pi$.
\end{assumption}

\begin{remark}
\label{rem1}
Assumptions~\ref{3A1}-\ref{3A2} are standard requirements. Assumption~\ref{3A3}, on the other hand, is a restrictive requirement though often used in the literature, see for instance, \cite{Maei_PhD, FastGradient, GTD, dalal2018finite, tale2TS}, where this assumption has been made. This can however be relaxed if our analysis of $N$-timescale SA algorithms is extended to the case when the noise can have a general iterate-dependent Markovian structure that can appear in each of the $N$ recursions. This setting is popularly referred to as the Markov noise setting. There is no prior work on stability of multi-scale algorithms with Markov noise even though \cite{Arunselvan1} does present stability conditions for one-timescale algorithms with Markov noise. We say more on this in Section~\ref{conclusion}.
\end{remark}


\begin{theorem}
    \label{theorem_GTD2-M_3TS_app}
     Suppose Assumptions \ref{3A1}, \ref{3A2} and \ref{3A3} hold and let $w>0$. Then, the GTD2-M-3TS iterates given by \eqref{GTD2-M-1_app}-\eqref{GTD2-M-3_app} satisfy \(\theta_{t} \rightarrow \theta^{*} = -\bar{A}^{-1}\bar{b}\) a.s. as \(t\rightarrow \infty\). 
\end{theorem}
\emph{Proof.} 
We transform \eqref{GTD2-M-1_app}-\eqref{GTD2-M-3_app} into the standard form. Let $\mathcal{F}_t = \sigma(u_0, v_0, \theta_0, r_{j+1},\phi_j, \phi_j': j < t)$. Let, $A_t = \phi_t(\gamma\phi_t'-\phi_t)^T$ and $b_t = r_{t+1}\phi_t$. Then, \eqref{GTD2-M-1_app} can be rewritten as:
\begin{gather*}
        v_{t+1} = v_t + \xi_{t}\left(h^{(1)}(v_t,u_t,\theta_t) + M_{t+1}^{(1)}\right)\\
        \text{where, }\; h^{(1)}(v_t,u_t,\theta_t) = \mathbb{E}[(\phi_t - \gamma\phi_t')\phi_t^Tu_t - w v_t|\mathcal{F}_t] 
        = -\bar{A}^Tu_t-wv_t, \\
        M_{t+1}^{(1)} = -A_t^Tu_t - w v_t - h(v_t,u_t,\theta_t)  = (\bar{A}^T-A_t^T)u_t.
\end{gather*}
Next, \eqref{GTD2-M-2_app} can be re-written as:
\begin{gather*}
        u_{t+1} = u_t + \beta_t\left(h^{(2)}(v_t,u_t,\theta_t) + M_{t+1}^{(2)}\right),\\
        \text{where,}\; h^{(2)}(v_t,u_t,\theta_t) = \mathbb{E}[\delta_t\phi_t - \phi_t\phi_t^T u_t|\mathcal{F}_t]
        = \bar{A}\theta_t + \bar{b} - \bar{C}u_t,\\
        M_{t+1}^{(2)} = A_t\theta_t + b_t - C_t u_t - g(v_t,u_t,\theta_t),= (A_t-\bar{A})\theta_t + (b_t-\bar{b}) + (\bar{C} - C_t)u_t.
\end{gather*}
Here, $C_t = \phi_t\phi_t^T$ and $\bar{C} = \mathbb{E}[\phi_t\phi_t^T]$. Finally, \eqref{GTD2-M-3_app} can be re-written as:
\begin{align*}
        \theta_{t+1} = \theta_t + \varrho_{t}\left(h^{(3)}(v_t,u_t,\theta_t) + \varepsilon_t + M_{t+1}^{(3)}\right),
\end{align*}
where
\(
        h^{(3)}(v_t,u_t,\theta_t) = v_t \mbox{ and }
        M_{t+1}^{(3)} = 0.
\)
We show that conditions \textbf{(A:1)}-\textbf{(A:3)}, \textbf{(B.3.1)}-\textbf{(B.3.3)} and \textbf{(C.3.1)}-\textbf{(C.3.3)} hold. 
The functions $h^{(1)},h^{(2)},h^{(3)}$ are linear in $v,u,\theta$ and hence Lipschitz continuous, thereby satisfying (\textbf{A:1}). We choose the step-size sequences such that they satisfy (\textbf{A:2}). One popular choice is 
\begin{equation}
\xi_t = \frac{1}{(t+1)^{\xi}}, \beta_t = \frac{1}{(t+1)^{\beta}}, \varrho_t= \frac{1}{(t+1)^{\varrho}},
\end{equation}
with \({\displaystyle \frac{1}{2}<\xi<\beta<\varrho\leq1}\).
Now, $M_{t+1}^{(1)},M_{t+1}^{(2)}$ and $M_{t+1}^{(3)}$ $t\geq0$, are martingale difference sequences w.r.t $\mathcal{F}_t$ by construction. Further, \[\mathbb{E}[\|M_{t+1}^{(1)}\|^2|\mathcal{F}_t] \leq \|(\bar{A}^T - A_t^T)\|^2 \|u_t\|^2,\] \[\mathbb{E}[\|M_{t+1}^{(2)}\|^2|\mathcal{F}_t] \leq 3(\|(A_t-\bar{A})\|^2 \|\theta_t\|^2 + \|(b_t-\bar{b})\|^2 + \|(\bar{C}-C_t)\|^2\|u_t\|^2).\] 
Note that (\textbf{A:3}) is satisfied with \(K_1 = \|(\bar{A}^T-A_{t}^T)\|^2,\) 
\(K_2 = 3\max(\|A_t - \bar{A}\|^2,\|b_t-\bar{b}\|^2,\|(\bar{C}-C_t)\|^2),\)
and any $K_3>0$.
The fact that $K_1,K_2<\infty$ follows from Assumption \ref{3A1}. 
For a fixed $u,\theta \in \mathbb{R}^d$, consider the ODE
\[\dot{v}(t) = -\bar{A}^Tu - w v(t).\]
For $w>0$, $\lambda^{(1)}(u,\theta) = -\frac{\bar{A}^Tu}{w}$ is the unique globally asymptotically stable equilibrium (g.a.s.e), is linear and therefore Lipschitz continuous. This satisfies \textbf{(C.3.1)}. 
Next, for a fixed $\theta \in \mathbb{R}^d$,
\[\dot{u}(t) = \bar{A}\theta + \bar{b} -\bar{C}u(t),\]
has $\lambda^{(2)}(\theta) = \bar{C}^{-1}(\bar{A}\theta + \bar{b})$ as its unique g.a.s.e because $-\bar{C}^{-1}$ is negative definite. Also $\lambda^{(2)}(\theta)$ is linear in $\theta$ and therefore Lipschitz. This satisfies \textbf{(C.3.2)}. 
Finally, to satisfy \textbf{(C.3.3)}, consider,
\begin{align*}
        \dot{\theta}(t)
       = \frac{-\bar{A}^T\bar{C}^{-1}\bar{A}\theta(t)-\bar{A}^T\bar{C}^{-1}\bar{b}}{w}.
\end{align*}
Since $\bar{A}$ is negative definite and $\bar{C}$ is positive definite,  $-\bar{A}^T\bar{C}^{-1}\bar{A}$ is negative definite as well. Therefore, $\theta^* = -\bar{A}^{-1}\bar{b}$ is the unique g.a.s.e for the above ODE.

Next, we show that the sufficient conditions for stability of the three iterates are satisfied. The function, $h_c^{(1)}(v,u,\theta) = \frac{-c\bar{A}^Tu-wcv}{c} = -\bar{A}^Tu-wv \rightarrow h_{\infty}^{(1)}(v,u,\theta) = -\bar{A}^Tu-wv$ uniformly on compacts as $c\rightarrow\infty$. The limiting ODE: 
\[\dot{v}(t) = -\bar{A}^Tu-wv(t)\]
has $\lambda_{\infty}^{(1)}(u,\theta) = -\frac{\bar{A}^Tu}{w}$ as its unique g.a.s.e. $\lambda^{(1)}_{\infty}$ is Lipschitz with $\lambda_{\infty}^{(1)}(0,0) = 0$, thus satisfying assumption \textbf{(B.3.1)}
The function, $h_c^{(2)}(u,\theta) = \frac{c\bar{A}\theta + \bar{b} - c\bar{C}u}{c} = \bar{A}\theta-\bar{C}u+\frac{\bar{b}}{c} \rightarrow h_{\infty}^{(2)}(u,\theta) = \bar{A}\theta - \bar{C}u$ uniformly on compacts as $c\rightarrow\infty$. The limiting ODE \[\dot{u}(t) = \bar{A}\theta - \bar{C}u(t)\]
has $\lambda^{(2)}_{\infty}(\theta) = \bar{C}^{-1}\bar{A}\theta$ as its unique g.a.s.e. since $-\bar{C}$ is negative definite. $\lambda^{(2)}_{\infty}$ is Lipschitz with $\lambda^{(2)}_{\infty}(0) = 0$. Thus assumption \textbf{(B.3.2)} is satisfied.

Finally, $h_c^{(3)}(\theta) = \frac{-c\bar{A}^T\bar{C}^{-1}\bar{A}\theta}{cw} \rightarrow h^{(3)}_{\infty}(\theta) = \frac{-\bar{A}^T\bar{C}^{-1}\bar{A}\theta}{w}$ uniformly on compacts as $c\rightarrow \infty$ and the ODE:
\[\dot{\theta}(t) = -\frac{\bar{A}^T\bar{C}^{-1}\bar{A}\theta(t)}{w}\]
has the origin in $\mathbb{R}^d$ as its unique g.a.s.e. This ensures the final condition \textbf{(B.3.3)}. Further, observe that $\|\varepsilon_t^{(3)}\|=\xi_t\|\left((\phi_t - \gamma\phi_t')\phi_t^Tu_t - w v_t\right)\|\rightarrow0$ since $\xi_t\rightarrow0 \mbox{ as } t\rightarrow\infty$. By Theorem \ref{thm_with_pert}, 
\[
\begin{pmatrix}
    v_t\\
    u_t\\
    \theta_t
\end{pmatrix}
\rightarrow
\begin{pmatrix}
    \lambda(\Gamma(-\bar{A}^{-1}\bar{b}),-\bar{A}^{-1}\bar{b})\\
    \Gamma(-\bar{A}^{-1}\bar{b})\\
    -\bar{A}^{-1}\bar{b}.
\end{pmatrix}
=
\begin{pmatrix}
    0\\
    0\\
    -\bar{A}^{-1}\bar{b}.
\end{pmatrix}
\]
Specifically, $\theta_t \rightarrow -\bar{A}^{-1}\bar{b}$.
\hfill $\square$

\begin{theorem}
    \label{theorem_TDC-M_3TS_app}
      Suppose Assumptions \ref{3A1}, \ref{3A2} and \ref{3A3} hold and let $w>0$. Then, the TDC-M-3TS iterates given by \eqref{TDC-M-1_app}--\eqref{TDC-M-3_app} satisfy \(\theta_{t} \rightarrow \theta^{*} = -\bar{A}^{-1}\bar{b}\) a.s. as \(t\rightarrow \infty\). 
\end{theorem}
\emph{Proof.} 
As before, we transform the iterates given by \eqref{TDC-M-1_app}, \eqref{TDC-M-2_app} and \eqref{TDC-M-3_app} into the standard SA form. Let $\mathcal{F}_t = \sigma(u_0, v_0, \theta_0, r_{j+1},\phi_j, \phi_j': j < t)$. Let $A_t = \phi_t(\gamma\phi_t'-\phi_t)^T$ and $b_t = r_{t+1}\phi_t$. Then, \eqref{TDC-M-1_app} can be re-written as:
\begin{equation*}
    v_{t+1} = v_t + \xi_{t}\left(h^{(1)}(v_t,u_t,\theta_t) + M_{t+1}^{(1)}\right),
\end{equation*}
where $h^{(1)}(v_t,u_t,\theta_t) = \mathbb{E}[\delta_t\phi_t - \gamma\phi_{t}'\phi_{t}^{T}u_{t} - w v_{t}|\mathcal{F}_t] = \bar{A}\theta_t + \bar{b} -\gamma\mathbb{E}[\phi_t'\phi_t^T]u_t-wv_t$ and $M_{t+1}^{(1)} = \delta_t\phi_t - \gamma\phi_{t}'\phi_{t}^{T}u_{t} - w v_{t} - h^{(1)}(v_t,u_t,\theta_t) = (A_t-\bar{A})\theta_t + (b_t-\bar{b}) + \gamma(\mathbb{E}[\phi_t'\phi_t^T] - \phi_t'\phi_t^T)u_t.$
Next, \eqref{TDC-M-2_app} can be re-written as:
\begin{equation*}
        u_{t+1} = u_t + \beta_t\left(h^{(2)}(v_t,u_t,\theta_t) + M_{t+1}^{(2)}\right)
\end{equation*}
where, $h^{(2)}(v_t,u_t,\theta_t) = \mathbb{E}[\delta_t\phi_t - \phi_t\phi_t^T u_t|\mathcal{F}_t] = \bar{A}\theta_t + \bar{b} - \bar{C}u_t$ and $M_{t+1}^{(2)} = A_t\theta_t + b_t - C_t u_t - h^{(2)}(v_t,u_t,\theta_t) = (A_t-\bar{A})\theta_t + (b_t-\bar{b}) + (\bar{C} - C_t)u_t.$
Here, $C_t = \phi_t\phi_t^T$, $\bar{C} = \mathbb{E}[\phi_t\phi_t^T]$. Finally, \eqref{TDC-M-3_app} can be re-written as:
\[
        \theta_{t+1} = \theta_t + \varrho_{t}\left(h^{(3)}(v_t,u_t,\theta_t) + \varepsilon_t + M_{t+1}^{(3)}\right),
\]
where
\(
        h^{(3)}(v_t,u_t,\theta_t) = v_t \mbox{ and }
        M_{t+1}^{(3)} = 0.
\)
As before we show that the conditions \textbf{(A:1)}-\textbf{(A:3)}, \textbf{(B.3.1)}-\textbf{(B.3.3)} and \textbf{(C.3.1)}-\textbf{(C.3.3)} hold. The functions $h^{(1)},h^{(2)},h^{(3)}$ are linear in $v,u,\theta$ and hence Lipschitz continuous, therefore satisfying (\textbf{A:1}). We choose the step-size sequences such that they satisfy (\textbf{A:2}). One popular choice is \eqref{sseq}.
Observe now that $M_{t+1}^{(1)},M_{t+1}^{(2)}$ and $M_{t+1}^{(3)}$ $t\geq0$, are martingale difference sequences w.r.t $\mathcal{F}_t$ by construction. Next, \[\mathbb{E}[\|M_{t+1}^{(1)}\|^2|\mathcal{F}_t] \leq 3(\|(A_t-\bar{A})\|^2\|\theta_t\|^2 + \|(b_t-\bar{b})\|^2  +\gamma(\|\mathbb{E}[\phi_t^{'}\phi_t^T]-\phi_t'\phi_t^T\|^2)\|u_t\|^2),\]
\[ \mathbb{E}[\|M_{t+1}^{(2)}\|^2|\mathcal{F}_t] \leq 3(\|(A_t-\bar{A})\|^2 \|\theta_t\|^2 + \|(b_t-\bar{b})\|^2 + \|(\bar{C}-C_t)\|^2\|u_t\|^2).\]
The first part of (\textbf{A:3}) is satisfied with $K_1 =3 \max(\|(A_t-\bar{A})\|^2, \|(b_t-\bar{b})\|^2, \gamma(\|\mathbb{E}[\phi_t^{'}\phi_t^T]-\phi_t'\phi_t^T\|^2))$, $K_2 = 3\max(\|A_t - \bar{A}\|^2,\|b_t-\bar{b}\|^2,\|(\bar{C}-C_t)\|^2)$ and any $K_3>0$.
The fact that $K_1,K_2<\infty$ follows from Assumption \ref{3A1}. For a fixed $u,\theta \in \mathbb{R}^d$, consider the ODE
\[\dot{v}(t) = \bar{A}\theta + \bar{b} -\gamma\mathbb{E}[\phi_t'\phi_t^T]u-wv(t).\]
For $w>0$, $\lambda^{(1)}(u,\theta) = \frac{\bar{A}\theta+\bar{b}-\gamma\mathbb{E}[\phi_t'\phi_t^T]u}{w}$ is the unique g.a.s.e, is linear and therefore Lipschitz continuous. This satisfies \textbf{(C.3.1)}. Next, for a fixed $\theta \in \mathbb{R}^d$,
\[\dot{u}(t) = \bar{A}\theta + \bar{b} -\bar{C}u(t),\]
has $\lambda^{(2)}(\theta) = \bar{C}^{-1}(\bar{A}\theta + \bar{b})$ as its unique g.a.s.e because $-\bar{C}^{-1}$ is negative definite. Also $\lambda^{(2)}(\theta)$ is linear in $\theta$ and therefore Lipschitz. This satisfies \textbf{(C.3.2)}. Finally, to satisfy \textbf{(C.3.3)}, consider,
\begin{align*}
        \dot{\theta}(t)
       & = \frac{(I-\gamma\mathbb{E}[\phi_t'\phi_t^T]\bar{C}^{-1})(\bar{A}\theta(t) + \bar{b})}{w}.
\end{align*}
Now, $(I-\gamma\mathbb{E}[\phi_t'\phi_t^T]\bar{C}^{-1})\bar{A}$ = $(\mathbb{E}[\phi_t\phi_t^T]-\gamma\mathbb{E}[\phi_t'\phi_t^T])\bar{C}^{-1}\bar{A} = \mathbb{E}[(\phi_t - \gamma\phi_t')\phi_t^T]\bar{C}^{-1}\bar{A} = -\bar{A}^T\bar{C}^{-1}\bar{A}$. Since, $\bar{A}$ is negative definite and $\bar{C}$ is positive definite, therefore $ -\bar{A}^T\bar{C}^{-1}\bar{A}$ is negative definite and hence the above ODE has $\theta^* = -\bar{A}^{-1}\bar{b}$ as its unique g.a.s.e. 

Next, we show that the sufficient conditions for stability of the three iterates are satisfied. The function, $h_c^{(1)}(v,u,\theta) = \frac{c\bar{A}\theta + \bar{b}-c\gamma\mathbb{E}[\phi_t'\phi_t^{T}]u-cwv}{c} = \bar{A}\theta+\bar{b}{c} -\gamma\mathbb{E}[\phi_t'\phi_t^{T}]u-wv \rightarrow h_{\infty}^{(1)}(v,u,\theta) = \bar{A}\theta-\gamma\mathbb{E}[\phi_t'\phi_t^{T}]u-wv$ uniformly on compacts as $c\rightarrow\infty$. The limiting ODE: 
\[\dot{v}(t) = \bar{A}\theta_t-\gamma\mathbb{E}[\phi_t'\phi_t^{T}]u_t-wv(t)\]
has $\lambda^{(1)}_{\infty}(u,\theta) = \frac{\bar{A}\theta-\gamma\mathbb{E}[\phi_t'\phi_t^{T}]u}{w}$ as its unique g.a.s.e. $\lambda^{(1)}_{\infty}$ is Lipschitz with $\lambda^{(1)}_{\infty}(0,0) = 0$, thus satisfying assumption \textbf{(B.3.1)}.

The function $h_c^{(2)}(u,\theta) = \frac{c\bar{A}\theta + \bar{b} - c\bar{C}u}{c} = \bar{A}\theta-\bar{C}u+\frac{\bar{b}}{c} \rightarrow h_{\infty}^{(2)}(u,\theta) = \bar{A}\theta - \bar{C}u$ uniformly on compacts as $c\rightarrow\infty$. The limiting ODE 
\[\dot{u}(t) = \bar{A}\theta - \bar{C}u(t)\]
has $\lambda^{(2)}_{\infty}(\theta) = \bar{C}^{-1}\bar{A}\theta$ as its unique g.a.s.e. since $-\bar{C}$ is negative definite. $\lambda^{(2)}_{\infty}$ is Lipschitz with $\lambda^{(2)}_{\infty}(0) = 0$. Thus assumption \textbf{(B.3.2)} is satisfied.

\noindent Finally, 
\[h_c^{(3)}(\theta) = \frac{c\bar{A}\theta-c\gamma\mathbb{E}[\phi_t'\phi_t^T]\bar{C}^{-1}\bar{A}\theta}{cw} \rightarrow h^{(3)}_{\infty} = \frac{(I-\gamma\mathbb{E}[\phi_t'\phi_t^T]\bar{C}^{-1})\bar{A}\theta}{w}\]
uniformly on compacts as $c\rightarrow \infty$. Consider the ODE:
\[\dot{\theta}(t) = \frac{(I-\gamma\mathbb{E}[\phi_t'\phi_t^T]\bar{C}^{-1})\bar{A}\theta(t)}{w}.\]

Now $(I-\gamma\mathbb{E}[\phi_t'\phi_t^T]\bar{C}^{-1})\bar{A}$ = $(\mathbb{E}[\phi_t\phi_t^T]-\gamma\mathbb{E}[\phi_t'\phi_t^T])\bar{C}^{-1}\bar{A} = \mathbb{E}[(\phi_t - \gamma\phi_t')\phi_t^T]\bar{C}^{-1}\bar{A} = -\bar{A}^T\bar{C}^{-1}\bar{A}$. Since $\bar{A}$ is negative definite and $\bar{C}$ is positive definite, $-\bar{A}^T\bar{C}^{-1}\bar{A}$ is negative definite and hence the above ODE has the origin as its unique g.a.s.e. This ensures the final condition \textbf{(B.3.3)}. Next, observe that $\|\varepsilon_t^{(3)}\|=\xi_t\|\left((\phi_t - \gamma\phi_t')\phi_t^Tu_t - w v_t\right)\|\rightarrow0$ since $\xi_t\rightarrow0 \mbox{ as } t\rightarrow\infty$. By Theorem \ref{thm_with_pert}, 
\[
\begin{pmatrix}
    v_t\\
    u_t\\
    \theta_t
\end{pmatrix}
\rightarrow
\begin{pmatrix}
    \lambda(\Gamma(-\bar{A}^{-1}\bar{b}),-\bar{A}^{-1}\bar{b})\\
    \Gamma(-\bar{A}^{-1}\bar{b})\\
    -\bar{A}^{-1}\bar{b}
\end{pmatrix}
=
\begin{pmatrix}
    0\\
    0\\
    -\bar{A}^{-1}\bar{b}
\end{pmatrix}
\]
Specifically, $\theta_t \rightarrow -\bar{A}^{-1}\bar{b}$ almost surely.
\hfill $\square$

\subsubsection{Experiments}
\hfill \break
The Gradient TD algorithms along with their momentum variants are evaluated on two standard MDPs: 5-State Random Walk \cite{FastGradient} and Boyan Chain \cite{Boyan}. Along with 3-timescale version of the algorithms we consider a 4-timescale version as defined by the following iterates:
\begin{enumerate}
    \item \textbf{GTD2-M-4TS}: 
    \begin{equation}
    \label{gtd2_M_1_4TS}
        \theta_{t+1} = \theta_{t} + \alpha_{t}(\phi_{t} - \gamma\phi_{t}')\phi_{t}^{T}u_{t} + \eta^{(1)}_{t}(\theta_{t} - \theta_{t-1}),
    \end{equation}
    \begin{equation}
    \label{gtd2_M_2_4TS}
         u_{t+1} = u_{t} + \beta_{t}(\delta_{t} - \phi_{t}^{T}u_{t})\phi_{t} + \eta^{(2)}_{t}(u_{t} - u_{t-1}).
    \end{equation}

    \item \textbf{TDC-M-4TS}:
    \begin{equation}
    \label{tdc_M_1_4TS}
        \theta_{t+1} = \theta_{t} + \alpha_{t}(\delta_{t}\phi_{t} - \gamma\phi_{t}'(\phi_{t}^{T}u_{t}))+ \eta^{(1)}_{t}(\theta_{t} - \theta_{t-1}),
    \end{equation}
    \begin{equation}
    \label{tdc_M_2_4TS}
         u_{t+1} = u_{t} + \beta_{t}(\delta_{t} - \phi_{t}^{T}u_{t})\phi_{t} + \eta^{(2)}_{t}(u_{t} - u_{t-1}).
    \end{equation}
\end{enumerate}
We consider decreasing step-sizes of the form:
\(\varrho_t^{(1)} = \frac{1}{(t+1)^{{\varrho}^{(1)}}}\),
\(\varrho_t^{(2)} = \frac{1}{(t+1)^{{\varrho}^{(2)}}}\),
\(\beta_t = \frac{1}{(t+1)^{\beta}}\), 
\(\alpha_t = \frac{1}{(t+1)^{\alpha}}\), respectively,
in all the examples. In the 3-Timescale case the conditions on step size turn out to be $\alpha<\varrho+\beta,$ and $\beta<\varrho$, while in the 4-Timescale case the conditions are $\alpha<\beta+\varrho^{(1)}-\varrho^{(2)},$ $\beta<2\varrho^{(2)}$ and $\varrho^{(2)}<\varrho^{(1)}$. 
Our analysis of convergence requires square-summability of step-size sequences. However, such a choice is seen to slow down the convergence of the algorithm.
Recently, \cite{dalal2018finite} provided convergence rate results for Gradient TD schemes with non-square-summable step-sizes (See Remark 2 of \cite{dalal2018finite}).
Motivated by this, we look at non-square summable step-sizes for our experiments, and observe that the iterates empirically converge in such cases as well.

For a detailed description of the MDPs, see Figure \ref{f-5state} and \ref{f-boyan}. See Figure \ref{f1} for the results on the 5-State Random Walk and Figure \ref{f2} for results on the Boyan Chain. The exact values of $\varrho^{(1)},\varrho^{(2)},\alpha$ and $\beta$ are provided in Table \ref{stepsize-table} and $w=0.1$ for both 3-TS and 4-TS settings. Our results indicate that adding momentum terms to the algorithms clearly improves performance over their vanilla counterparts. Nonetheless, it is less clear from these experiments as to whether adding momentum term to both the iterates is better as opposed to doing the same for only one of the iterates. In particular, for GTD2, the 4-TS scheme appears to perform better while for TDC, the 3-TS version is seen to be better. Further analysis of the finite time behaviour of these algorithms needs to be carried out in the future to better assess the performance of these algorithms.

\begin{figure*}[!bt]
    \centering
    {
    {\includegraphics[width=\linewidth]{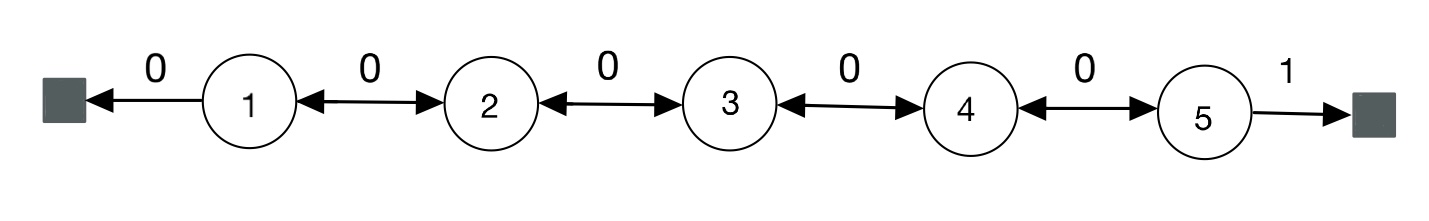}}
    
    \caption{5-State Random Walk from \cite{FastGradient}.}
    \label{f-5state}
    {\includegraphics[width=\linewidth]{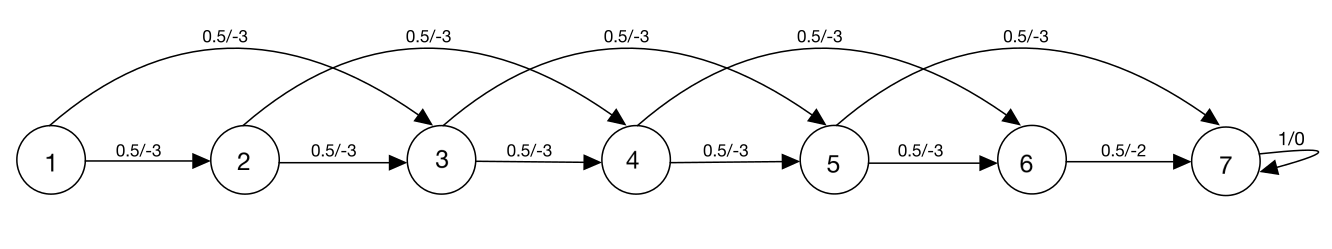}}
    \caption{7-state Boyan Chain from \cite{Boyan}}
    \label{f-boyan}
    }
\end{figure*}
\begin{table}
\caption{Choice of step-size parameters} 
\centering 
\begin{tabular}{c c c c c} 
\hline\hline 
5-State RW & $\alpha$ & $\beta$ & $\varrho^{(1)}$ &$\varrho^{(2)}$\\[0.5ex] 
\hline 
Vanilla & 0.4 & 0.4 & - & - \\ 
Three-TS & 0.4 & 0.4 & 0.5 & - \\
Four-TS & 0.4 & 0.4 & 0.5 & 0.25\\[1ex] 
\hline\hline 
Boyan Chain & $\alpha$ & $\beta$ & $\varrho^{(1)}$ &$\varrho^{(2)}$\\[0.5ex] 
\hline 
Vanilla & 0.4 & 0.4 & - & - \\ 
Three-TS & 0.35 & 0.35 & 0.45 & - \\
Four-TS & 0.35 & 0.35 & 0.45 & 0.35\\[1ex]
\hline\hline 
\end{tabular}
\label{stepsize-table} 
\end{table}
\begin{figure*}[!bt]
    \centering
    {
    {\includegraphics[width=0.8\linewidth]{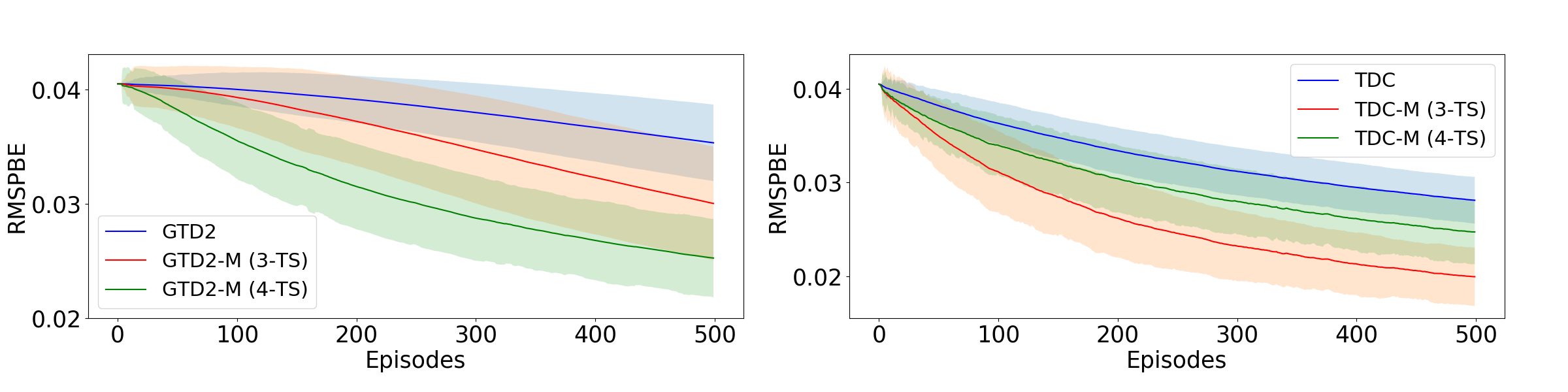}}
    
    \caption{RMSPBE across episodes (averaged over 100 independent runs) for 5-State Random Walk. The features used are the \textit{Dependent} features used in \cite{FastGradient}.}
    \label{f1}
    {\includegraphics[width=0.8\linewidth]{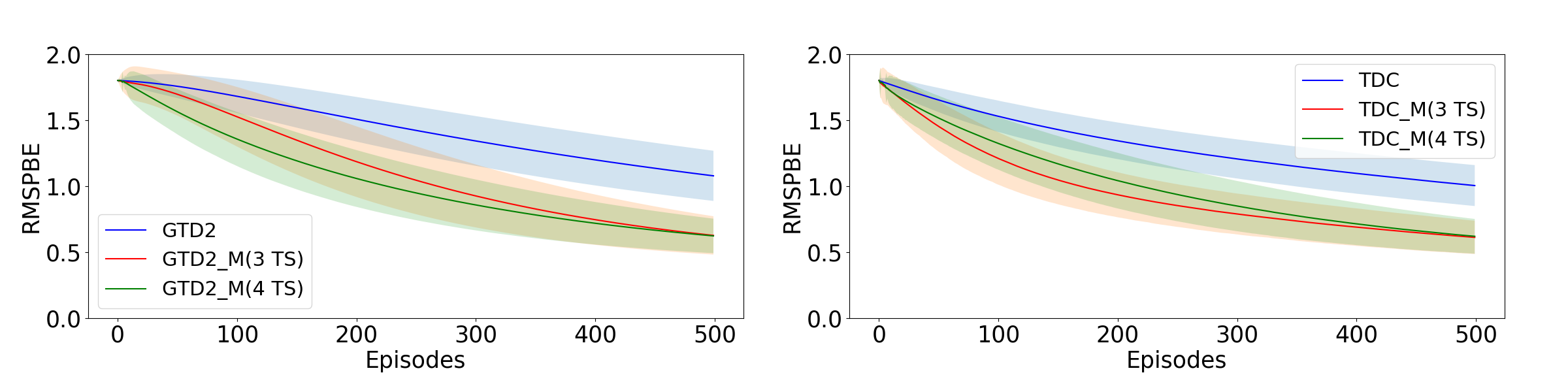}}
    \caption{RMSPBE (averaged over 100 independent runs) across episodes for the Boyan Chain problem. The features used are the standard spiked features of size 4 used in Boyan chain (see \cite{Dann}).}
    \label{f2}
    }
\end{figure*}

\subsection{Off-Policy Actor Critic}

We consider policy optimization with linear function approximation for policy evaluation. Here we adopt similar notations as before, where $\phi(s) \equiv \phi_{s}$ is a feature vector associated with the state $s$ and $\theta$ is the associated parameter vector. For ease of notation, we consider cost minimization, where we want to find the policy $\pi^*$ that minimizes the long term discounted cost function $J^\pi= \mathbb{E}_{s_0 \sim \xi}[\sum_{i=0}^{\infty}\gamma^i r_i\mid s_0]$, where $\xi$ denotes the initial state distribution. The TD-error is defined by \(\delta_{t} = r_{t+1} + \gamma \theta_{t}^{T}\phi_{t+1} - \theta_{t}^{T}\phi_{t}\). The update rules are given by:
\begin{gather}
\label{alg:off-policy-ac}
        u_{t+1} = u_{t} + \alpha_{t}\rho_t(\delta_{t} - \phi_{t}^{T}u_{t})\phi_{t}, \\
        \theta_{t+1} = \theta_{t} + \beta_{t}\rho_t\left(\delta_{t}\phi_{t} - \alpha_{t}\gamma \phi_{t}'(\phi_{t}^{T}u_{t})\right),\\
    z_{t+1} = z_{t} - \eta_{t}(\rho_t \delta_{t} \nabla_{z_{t}} \log \pi_{z_{t}}(a_{t}|s_{t})+\epsilon z_t),
\end{gather}
where $\rho_t = \frac{\pi_{z_t}(a_t|s_t)}{b(a_t|s_t)}$, where $b(a_t|s_t)$ is the behavior policy that is used to collect samples. Note that we use an additional term $\epsilon z_t$ term in our actor update, where $\epsilon >0 $ is a small constant. This is equivalent to adding a regularizing term $\epsilon \|z_t\|^2$ to the objective function. A different version of off-policy actor critic has been studied in \cite{degris-off-policy}, where an almost sure convergence result was also given. Their algorithm uses a one timescale update for the critic and uses a projection operator for the actor update to ensure stability of the iterates, whereas, we do not make use of projection operators for stability and instead show Assumptions \textbf{(B.3.1)}-\textbf{(B.3.3)} hold. The feature matrix $\Phi$ is an $n\times d$ matrix where the $s^{th}$ row is $\phi(s)^T$. In the following, we consider the i.i.d setting, where the tuple $(\phi_{t},\phi_{t}'$) (with $\phi_{t}' \equiv \phi_{t+1}$) is drawn independently from the stationary distribution $\{d_b(s)\}$. 

Additionally, we use the Gibbs (softmax) parameterization for the policy, defined as follows
\begin{align}
    \pi_z(a|s) = \frac{e^{\zeta(s,a)^Tz}}{\sum_{a'} e^{\zeta(s,a')^Tz}},
\end{align}
where $\zeta(s,a)$ is the state-action feature for the actor.
It can be seen that $\frac{\nabla_z \pi_z(a|s)}{\pi_z(a|s)} = \nabla_z \log \pi_z (a|s) = \zeta(s,a)- \sum_{a'}\zeta(s,a')\pi(a'|s)$. Let \(\bar{A}_z \stackrel{\triangle}{=} \mathbb{E}[\phi_{t}(\gamma\phi_{t}'-\phi_{t})^{T}]\) and \(\bar{b}_z \stackrel{\triangle}{=} \mathbb{E}[r_{t+1}\phi_t]\), where the expectations are w.r.t. the stationary distribution of the Markov chain induced by the policy $\pi_z$. The matrix $\bar{A}_z$ is known to be negative definite for all $z$ (see \cite{Maei_PhD,TsitsiklisVanRoy}). As in previous sections, we make the following standard assumptions:

\begin{assumption}
\label{AC_A1}
All rewards $r(s,s')$, critic features $\phi(s)$ and actor features $\zeta(s,a)$ are bounded, i.e., $r(s,s')\leq 1$, $\|\phi(s)\|\leq 1$ and $\|\zeta(s,a)\|\leq 1$ $\forall s,s' \in \mathcal{S}$ and $a \in \mathcal{A}$. Also, the matrix $\Phi$ has full rank, where $\Phi$ is an $n\times d$ matrix where the s$^{th}$ row is $\phi(s)^T$.
\end{assumption}
\begin{assumption}
\label{AC_A4}
    The behavior policy has a minimum positive value $b_{\min}\in(0,1]: b(a|s)\geq b_{\min} \forall (s,a) \in \mathcal{S} \times \mathcal{A}$.
\end{assumption}
\begin{assumption}
\label{AC_A3}
The samples ($\phi_{t},\phi_{t}'$) are drawn i.i.d from the stationary distribution of the Markov chain induced by the behaviour policy $d_b$.
\end{assumption}
Let $\mathcal{F}_t = \sigma(z_0, u_0, \theta_0, r_{j+1},\phi_j, \phi_j': j < t)$. Let, $A_t = \phi_t(\gamma\phi_t'-\phi_t)^T$, $b_t = r_{t+1}\phi_t$ and $C_t = \phi_t \phi_t^T$.
We can rewrite the updates in the following multi-timescale SA form:
\begin{gather}
    u_{t+1} = u_{t} + \alpha_{t} (h^{(1)}(u_{t},\theta_{t},z_{t})+M_{t+1}^{(1)})\\
    \theta_{t+1} = \theta_{t} + \beta_{t} (h^{(2)}(u_{t},\theta_{t},z_{t})+M_{t+1}^{(2)})\\
    z_{t+1} = z_{t} + \eta_{t}(h^{(3)}(u_{t},\theta_{t},z_{t})+M_{t+1}^{(3)}).\
\end{gather}
Here, $h^{(1)}(u,\theta,z) = \mathbb{E}_{d_b}[\rho_t \delta_t\phi_t - \phi_t\phi_t^T u_t|\mathcal{F}_t]
= \bar{A}\theta_t + \bar{b} - \bar{C}u_t,$ where $\bar{A}_t \equiv \bar{A}_{z_t} = \mathbb{E}_{d_{\pi_t}}[A_t]$, $\bar{b}_t \equiv \bar{b}_{z_t} = \mathbb{E}_{d_{\pi_t}}[b_t]$, $\bar{C}_t \equiv \mathbb{E}_{d_{b}}[C_t]$ and $\pi_t = \pi_{z_{t}}$. Further $
M_{t+1}^{(1)} = (A_t-\bar{A})\theta_t + (b_t-\bar{b}) + (\bar{C} - C_t)u_t$.

Similarly,
$h^{(2)}(u_{t},\theta_{t},z_{t}) = \mathbb{E}_{d_{b}}[\rho_t(\delta_t\phi_t - \gamma\phi_{t}'\phi_{t}^{T}u_{t})|\mathcal{F}_t] = \bar{A}_t\theta_t + \bar{b}_t -\gamma\mathbb{E}_{d_{\pi_t}}[\phi_t'\phi_t^T]u_t$,  $M_{t+1}^{(2)} = (A_t-\bar{A}_t)\theta_t + (b_t-\bar{b}_t) + \gamma(\mathbb{E}_{d_{\pi_t}}[\phi_t'\phi_t^T] - \phi_t'\phi_t^T)u_t$, $h^{(3)}(u,\theta,z) = -\E_{d_b}\left[\rho_{t}\delta_{t}\nabla_{z_{t}} \log \pi_{z_{t}}(a_{t}|s_{t})| \mathcal{F}_{t}\right]-\epsilon z_t$, $M_{t+1}^{(3)} = \rho_{t}\delta_{t}\nabla_{z_{t}} \log \pi_{z_{t}}(a_{t}|s_{t}) - \E\left[\rho_{t}\delta_{t}\nabla_{z_{t}} \log \pi_{z_{t}}(a_{t}|s_{t})| \mathcal{F}_{t}\right]$.

Note that for fixed $z_t$, $h^{(1)}(u_{t},\theta_{t},z_{t})$ and $h^{(2)}(u_{t},\theta_{t},z_{t})$ are linear in $\theta_{t},u_{t}$ and therefore Lipschitz. Separately, note that $A_{z_t} = \sum_{s\in \mathcal{S}} d^{\pi_{z_t}}(s) \sum_{a \in \mathcal{A}} \pi_{z_t}(a|s) A_t$ and $\nabla_z A_{z_t} = \sum_{s\in \mathcal{S}} \sum_{a \in \mathcal{A}} (\nabla_z d^{\pi_{z_t}}(s)  \pi_{z_t}(a|s) A_t+d^{\pi_{z_t}}(s) \nabla_z \pi_{z_t}(a|s) A_t)$. It is known that $\nabla_z d^{\pi_{z_t}}(s)$ exists and is bounded \cite{Schweitzer_1968} while $\nabla_z \pi_z(a|s)$ and $A_t$ are bounded due to Assumption \ref{AC_A1}. Thus, derivative of $A_z$ with respect to $z$ is bounded and it follows that $h^{(1)}(\theta_1,u_t,z_t) = A_{z_t} u_t$ is Lipschitz in $z_t$ for fixed $u_t$. Similar arguments show that $h^{(2)}(\theta_1,u_t,z_t)$ is Lipschitz in $z_t$. To see that $h^{(3)}(\theta_1,u_t,z_t)$ is Lipschitz in $z_t$, note that $\nabla_z (\nabla_z \log \pi_z(a|s)) = - \sum_{a'}\zeta(s,a')\nabla_z \pi(a'|s)$, which exists and is bounded. 

We choose the step-size sequences such that they satisfy (\textbf{A:2}). One popular choice is 
\begin{equation}
\label{sseq}
\alpha_t = \frac{1}{(t+1)^{\alpha}}, \;\; \beta_t = \frac{1}{(t+1)^{\beta}},\;\; \eta_t= \frac{1}{(t+1)^{\eta}},
\end{equation}
with \({\textstyle \frac{1}{2}<\alpha<\beta<\eta\leq 1}\).
Now, $M_{t+1}^{(1)},M_{t+1}^{(2)}$ and $M_{t+1}^{(3)}$ $t\geq0$, are martingale difference sequences w.r.t $\mathcal{F}_t$ by construction. Further, \[\mathbb{E}[\|M_{t+1}^{(1)}\|^2|\mathcal{F}_t] \leq \|\bar{A}_z - A_t\|^2 \|\theta_t\|^2+\|b_t-\bar{b}_z\|^2 +\|\bar{C}_z - C_t\|^2 \|u_t\|^2,\] \[\mathbb{E}[\|M_{t+1}^{(2)}\|^2|\mathcal{F}_t] \leq 3(\|(A_t-\bar{A})\|^2 \|\theta_t\|^2 + \|(b_t-\bar{b})\|^2 + \|\gamma \mathbb{E}_{d_{\pi_t}}[\phi_t'\phi_t^T]-\gamma\rho_t \phi_t'\phi_t^T\|^2\|u_t\|^2),\] 
\[\mathbb{E}[\|M_{t+1}^{(3)}\|^2|\mathcal{F}_t] \leq 3(\|(D_t-\bar{D})\|^2 \|\theta_t\|^2+\|\rho_t r_t\nabla \log \pi_{z_t} (a_t|s_t)-\mathbb{E}_{d_b}[\rho_t r_t\nabla \log \pi_{z_t} (a_t|s_t)]\|^2,\]
where $D_t = \rho_t (\gamma \phi_t'-\phi_t) \nabla \log \pi_{z_t} (a_t|s_t)$ and $\bar{D} = \mathbb{E}_{d_b}[D_t]$. Note that (\textbf{A:3}) is satisfied with $K_1  = 3\max(\|A_t - \bar{A}\|^2,\|b_t-\bar{b}\|^2,\|(\bar{C}-C_t)\|^2)$, $K_2  = 3\max(\|A_t - \bar{A}\|^2,\|b_t-\bar{b}\|^2,\|\gamma \mathbb{E}_{d_{\pi_t}}[\phi_t'\phi_t^T]-\gamma\rho_t \phi_t'\phi_t^T\|^2)$
and $K_3 = 2 \max (\|D_t - \bar{D}\|^2,\|\rho_t r_t\nabla \log \pi_{z_t} (a_t|s_t)-\mathbb{E}_{d_b}[\rho_t r_t\nabla \log \pi_{z_t} (a_t|s_t)]\|^2)$.
The fact that $K_1,K_2,K_3<\infty$ follows from Assumption \ref{AC_A1} and \ref{AC_A4}.

For a fixed $z$, define $a^* = \arg \max_a \zeta(s,a)^Tz$  for each state $s$, and let $\pi_z^{\infty}(a|s) = 1_{a=a^*}$. It is known that $\pi_{cz}(a|s) \rightarrow \pi_z^{\infty}(a|s)$ uniformly on compacts \citet{chandru-SB}. Now 
\begin{align}
    h_c^{(1)}(u,\theta,z)  &=  \frac{h^{(1)}(cu,c\theta,cz)}{c} = \frac{\bar{A}_{cz}(c\theta) + \bar{b}_{cz} -\bar{C}(cu)}{c} = \bar{A}_{cz}\theta + \frac{\bar{b}_{cz}}{c} -\bar{C}u.
\end{align}
Let $A^{\infty}_z = \sum_s d^{\pi_{z}^{\infty}}(s)\sum_a \pi_{z}^{\infty}(a|s) A_t $. Note that $\lim_{c \rightarrow \infty} \bar{b}_{cz}/c = 0$ and $\bar{A}_{cz} \rightarrow \bar{A}^{\infty}_z $ uniformly on compacts. Since $-\bar{C}$ is negative-definite, the ODE $\dot{u}(t) = h^{(1)}_{\infty}(u(t),\theta,z)$ has $\lambda^{(1)}_{\infty}(\theta,z)=\bar{C}^{-1}\bar{A}_{z}^{\infty}\theta$ as its unique globally asymptotic stable equilibrium point. Further, $\lambda^{(1)}_{\infty}(\theta,z)$ is Lipschitz in $u$ and $z$ and $\lambda^{(1)}_{\infty}(0,0)=0$. Similarly,
\begin{align}
    h_c^{(2)}(\lambda^{(1)}_{\infty}(\theta,z),\theta,z)  &=  \frac{h^{(2)}(c\lambda^{(1)}_{\infty}(\theta,z),c\theta,cz)}{c} = \frac{\bar{A}_{cz}(c\theta) + \bar{b}_{cz} -\gamma\mathbb{E}_{d_{\pi_t}}[\phi_t'\phi_t^T] (\bar{C}^{-1}\bar{A}_{z}^{\infty}(c\theta))}{c} \\
    &= (\bar{A}_{cz}  -\gamma\mathbb{E}_{d_{\pi_t}}[\phi_t'\phi_t^T] \bar{C}^{-1}\bar{A}_{z}^{\infty}) \theta+ \frac{\bar{b}_{cz}}{c}.
\end{align}

Again, since $\lim_{c \rightarrow \infty} \bar{b}_{cz}/c = 0$, $\bar{A}_{cz} \rightarrow \bar{A}^{\infty}_z $ and $\mathbb{E}_{cz}[\phi_t'\phi_t^T] \rightarrow \mathbb{E}_{\pi_{z}^{\infty}}[\phi_t'\phi_t^T]$ uniformly on compacts, it follows that $h_c^{(2)}(\lambda^{(1)}(\theta,z),\theta,z) \rightarrow h_{\infty}^{(2)}(\lambda^{(1)}(\theta,z),\theta,z) \equiv (I -\gamma\mathbb{E}_{\pi_{z}^{\infty}}[\phi_t'\phi_t^T] \bar{C}^{-1})A^{\infty}_z \theta$ uniformly on compacts.  Now $(I-\gamma\mathbb{E}[\phi_t'\phi_t^T]\bar{C}_z^{-1})\bar{A}^{\infty}_z$ = $(\mathbb{E}[\phi_t\phi_t^T]-\gamma\mathbb{E}[\phi_t'\phi_t^T])\bar{C}_z^{-1}\bar{A}^{\infty}_z = \mathbb{E}[(\phi_t - \gamma\phi_t')\phi_t^T]\bar{C}_z^{-1}\bar{A}^{\infty}_z= -(\bar{A}^{\infty}_z)^T\bar{C}_z^{-1}\bar{A}^{\infty}_z$. Since $\bar{A}_z^{\infty}$ is negative definite and $\bar{C}_z^{\infty}$ is positive definite, $-(\bar{A}_z^{\infty})^T(\bar{C}_z^{\infty})^{-1}\bar{A}_z^{\infty}$ is negative definite and hence the above ODE has the origin as its unique g.a.s.e. Thus, $\lambda^{(2)}_{\infty}(z) = 0$ is Lipschitz in $z$ and $\lambda^{(2)}_{\infty}(0)=0$. Since $\lambda^{(2)}_{\infty}(z)=0$, we have
\begin{align}
\begin{split}
    h_c^{(3)}(\lambda^{(1)}_{\infty}(\lambda^{(2)}_{\infty}(z),z),\lambda^{(2)}_{\infty}(z),z)  &=  \frac{h^{(3)}(c\lambda^{(1)}_{\infty}(\lambda^{(2)}_{\infty}(z),z),c\lambda^{(2)}_{\infty}(z),cz)}{c} \\
    &= \frac{-\mathbb{E}_{cz}[r(s,s')\nabla \log \pi_{z}(a|s)]-c\epsilon z}{c}.
\end{split}
\end{align}
It can be seen that as $c \rightarrow \infty$, $h_c^{(3)}(\lambda^{(1)}_{\infty}(\lambda^{(2)}_{\infty}(z),z),\lambda^{(2)}_{\infty}(z),z) \rightarrow -\epsilon z$ uniformly on compacts. For all $\epsilon>0$, the ODE $\dot{z}(t) = h_{\infty}^{(3)}(\lambda^{(1)}_{\infty}(\lambda^{(2)}_{\infty}(z(t)),z(t)),\lambda^{(2)}_{\infty}(z(t)),z(t))$ has $0$ as its unique globally asymptotic stable equilibrium.

For a fixed $\theta,z \in \mathbb{R}^d$, consider the ODE
\(\dot{u}(t) = h^{(1)}(u(t),\theta,z)= \bar{A}_z \theta +\bar{b}_z - \bar{C}_z u(t).\)
It can be seen that $\lambda^{(1)}(\theta,z) = \bar{C}_z^{-1}(\bar{A}_z\theta+\bar{b}_z)$ is the unique g.a.s.e. This satisfies \textbf{(C.3.1)}. 

Next, for a fixed $z \in \mathbb{R}^d$,
\begin{align*}
    \dot{\theta}(t) &= h^{(2)}(\lambda^{(1)}(\theta(t),z),\theta(t),z) = \bar{A}_z\theta(t) + \bar{b}_z -\gamma \mathbb{E}_{d_{\pi_z}}[\phi'\phi^T]\bar{C}_z^{-1}(\bar{A}_z\theta(t)+\bar{b}_z)\\
    &= (\mathbb{E}_{d_{\pi_z}}[\phi \phi^T]-\gamma\mathbb{E}_{d_{\pi_z}}[\phi'\phi^T])\bar{C}_z^{-1}\bar{A}_z\theta(t) + (\mathbb{E}_{d_{\pi_z}}[\phi \phi^T]-\gamma \mathbb{E}_{d_{\pi_z}}[\phi'\phi^T])\bar{C}_z^{-1}\bar{b}_z\\
    &= -\bar{A}_z^T\bar{C}_z^{-1}\bar{A}_z\theta(t)  -\bar{A}_z^T\bar{C}_z^{-1}\bar{b}_z.
\end{align*}
Since $\bar{A}_z$ is negative definite and $\bar{C}_z$ is positive definite,  $-\bar{A}_z^T\bar{C}_z^{-1}\bar{A}_z$ is negative definite as well. Therefore, $\lambda^{(2)}(z) = -\bar{A}_z^{-1}\bar{b}_z$ is the unique g.a.s.e for the above ODE. This satisfies \textbf{(C.3.2)}.

Finally, consider the ODE
\begin{align*}
    \dot{z}(t) = h^{(3)}(\lambda^{(1)}(\lambda^{(2)}(z(t)),z(t)),\lambda^{(2)}(z(t)),z(t))=\Hat{g}(z(t))-\epsilon z(t),
\end{align*}
where $\hat{g}(z) = \E_{d_b}[\rho (r(s,s')+\gamma V(s')-V(s))\nabla \log \pi_z(a|s)]$. It is known that the ODE $\dot{z}(t)=\hat{g}(z(t))$ converges to the set $Z=\{z\mid\nabla_z J(z)=0\}$ \cite{degris-off-policy}. For small enough $\epsilon>0$, the ODE $\dot{z}(t) =\Hat{g}(z(t))-\epsilon z(t)$ will converge to an $\epsilon$-neighbourhood of $Z$. This can be seen by considering the Hartman-Grobman linearization of the ODE as in \cite{degris-off-policy}. Since the iterates $z_n$ track this ODE (see Lemma \ref{intermediate_lemma}$(ii)$), it follows that $z_n \rightarrow Z$ almost surely.

{\color{black} \subsection{Constrained Actor-Critic}
As in the preceding section we consider policy optimisation with linear function approximation for value estimation, retaining the notation $\phi(s)=\phi_{s}$ for the feature vector of state $s$. Our objective remains cost minimisation, now in the undiscounted average-reward setting: we seek a policy $\pi^{*}$ that minimises the long-run average cost $J^{\pi}:=\mathbb{E}_{s_{0}\sim\xi}\bigl[\sum_{i=0}^{\infty} r_{i}\mid s_{0}\bigr]$, where $\xi$ denotes the initial-state distribution. In contrast to the previous setup we impose $N$ additional constraints with per-step costs $g_{k}(s,a)$, $k=1,\dots,N$, requiring $g_{k}(s,a)\le \alpha_{k}$ for every $k$. Bhatnagar et al.~\cite{bhatnagar2012online} addressed this problem with an online three-timescale algorithm and proved almost-sure convergence; here we adapt their scheme by eliminating the projection in the actor update and adding a regulariser, as in the previous section, and we show that this modification still guarantees almost-sure convergence. Following \cite{bhatnagar2012online} we denote the actor feature vectors by $\psi_{s,a}$ and begin by stating the standard assumptions.

\begin{assumption}
\label{ACC_A1}
All costs $r(s,a)$, critic features $\phi_s$ and actor features $\psi_{s,a}$ are bounded, i.e., $r(s,a)\leq 1$, $\|\phi_s\|\leq 1$ and $\|\psi_{s,a}\|\leq 1$ $\forall s,s' \in \mathcal{S}$ and $a \in \mathcal{A}$. Also, the matrix $\Phi$ has full rank, where $\Phi$ is an $n\times d$ matrix where the s$^{th}$ row is $\phi(s)^T$.
\end{assumption}

\begin{assumption}
\label{ACC_A2}
The samples $\phi_{s_n}$ and $\phi_{s_{n}}'$ at time $n$ are drawn i.i.d from the stationary distribution of the Markov chain induced by the policy $\pi_{z(n)}$, $d^{\pi_{z(n)}}$.
\end{assumption}

With stationary distribution $d^{\pi_z}$ we define
\begin{align}
J(z)&=\sum_{s} d^{\pi_z}(s) \sum_{a\in A(s)} \pi_z(a| s)\, r(s,a),\\
G_k(z)&=\sum_{s} d^{\pi_z}(s) \sum_{a\in A(s)} \pi_z(a| s)\, g_k(s,a), \quad 1\le k\le N.
\end{align}
The constrained optimisation problem is
\[
\min_{z} J(z) \quad \text{s.t.}\; G_k(z) \le \alpha_k, \; 1\le k\le N.
\]
Introducing multipliers $\bar \gamma=(\gamma_1,\dots,\gamma_N)^{\top}\ge 0$, the Lagrangian becomes
\[
L(z,\bar \gamma)=J(z)+\sum_{k=1}^{N} \gamma_k\bigl(G_k(z)-\alpha_k\bigr).
\]
Its policy gradient can be written as
\[
\nabla_z L(z,\bar \gamma)=\mathbb{E}_{s\sim d^{\pi_z},\, a\sim \pi_z}\bigl[\nabla_z \log \pi_z(a| s) \, \mathcal{A}^{\pi_z,\bar \gamma}(s,a)\bigr],
\]
where the advantage is
\(
\mathcal{A}^{\pi_z,\bar \gamma}(s,a)=\mathcal{L}^{\pi_z,\bar \gamma}(s,a)-V^{\pi_z,\bar \gamma}(s).
\)
\begin{assumption}
\label{ACC:A3}
    Let stepsizes $a(n), b(n), c(n), d(n)$ satisfy the Robbins--Monro conditions
\begin{align}
 &\sum_{n} a(n)=\sum_{n} b(n)=\sum_{n} c(n)=\infty, 
 \quad\sum_{n}\bigl(a(n)^2+b(n)^2+c(n)^2\bigr)<\infty,\\
 &\lim_{n\to\infty}\frac{b(n)}{a(n)}=\lim_{n\to\infty}\frac{c(n)}{b(n)}=0\quad  \text{and}\quad d(n)=K a(n), \; K>0.
\end{align}
Denote $r_n:=r(s_n,a_n)$ and $g_{k,n}:=g_k(s_n,a_n)$.  The temporal‑difference error is
\[
\delta_n=r_n+\sum_{k} \gamma_k(n) \bigl(g_{k,n}-\alpha_k\bigr)-L_n+v(n)^{\top}(\phi_{s_{n+1}}-\phi_{s_n}).
\]
\end{assumption}

\noindent The recursions are then given by
\vspace{-2ex}
\begin{align}
L_{n+1}&=L_n+d(n)\Bigl(r_n+\sum_{k} \gamma_k(n)(g_{k,n}-\alpha_k)-L_n\Bigr),\\
v(n+1)&=v(n)+a(n)\,\delta_n\,\phi_{s_{n}},\\
z(n+1)&=z(n)-b(n)(\delta_n\,\psi_{s_n,a_n}+\epsilon z(n))\\
Y_k(n+1)&=Y_k(n)+a(n)\bigl(g_{k,n}-Y_k(n)\bigr),\\
\gamma_k(n+1)&=\Gamma\bigl(\gamma_k(n)+c(n)\bigl(Y_k(n)-\alpha_k\bigr)\bigr).
\end{align}
Note that the above algorithm is the same as in \cite{bhatnagar2012online} with an additional regularizer term for the policy update $z(n)$.
\noindent Fix $(z, \bar \gamma)$ and stack $\zeta(n)=(L_n, v(n)^{\top}, Y_1(n),\dots, Y_N(n))^{\top}\in\mathbb{R}^{1+d+N}$.  Then
\[
\zeta(n+1)=\zeta(n)+a(n)\bigl(A_n \zeta(n)+B_n\bar{\gamma}(n)+C_n\bigr),
\]
where 
\[
A_n=\begin{pmatrix}
-K & 0_{1\times d} & 0_{1\times N}\\[4pt]
-\phi_{s_n} & \phi_{s_n}(\phi_{s_{n+1}}-\phi_{s_n})^\top& 0_{d\times N}\\[4pt]
0_{N\times 1} & 0_{N\times d} & -I_N
\end{pmatrix},\qquad
C_n=\begin{pmatrix}
Kr_n\\[4pt]
 r_n \phi_{s_n}\\[4pt]
 g_n
\end{pmatrix},\qquad  B_n=\begin{pmatrix}
K(g_{k,n}-\alpha_k\bigr)^\top\\[4pt]
 \phi_{s_n}(g_{k,n}-\alpha_k\bigr)^\top\\[4pt]
 0_{N}^\top
\end{pmatrix}.
\]
Moreover,

\[
z(n+1)=z(n)+b(n)\bigl(-\epsilon_1 z(n)+X_n \zeta(n)+Y_n\bar{\gamma}(n)+Z_n\bigr),
\]
where 
\[
X_n=\begin{pmatrix}
-\psi_{s_n,a_n} ,& \psi_{s_n,a_n}(\phi_{s_{n+1}}-\phi_{s_n})^\top,& 0_{d\times N}
\end{pmatrix},\qquad
Z_n=r_n \psi_{s_n,a_n},\qquad  Y_n=\psi_{s_n,a_n}(g_{k,n}-\alpha_k\bigr)^\top.
\]

Thus, the algorithm can be written as the following three-timescale update:
\vspace{-2ex}
\begin{align}
\zeta(n+1)&=\zeta(n)+a(n)\bigl(A_n \zeta(n)+B_n\bar{\gamma}(n)+C_n\bigr)\\
z(n+1)&=z(n)+b(n)\bigl(-\epsilon z(n)+X_n \zeta(n)+Y_n\bar{\gamma}(n)+Z_n\bigr)\\
\gamma_k(n+1)&=\Gamma\bigl(\gamma_k(n)+c(n)\bigl(E_k\zeta(n)-\alpha_k\bigr)\bigr).
\end{align}

It can be rewritten as follows
\vspace{-2ex}
\begin{align}
\zeta(n+1)&=\zeta(n)+a(n)\bigl(A_n^z \zeta(n)+B_n^z\bar{\gamma}(n)+C_n^z+M^{(1)}(n+1)\bigr)\\
z(n+1)&=z(n)+b(n)\bigl(-\epsilon z(n)+X_n^z \zeta(n)+Y_n^z\bar{\gamma}(n)+Z_n^z+M^{(2)}(n+1)\bigr)\\
\gamma_k(n+1)&=\Gamma\bigl(\gamma_k(n)+c(n)\bigl(E_k\zeta(n)-\alpha_k\bigr)\bigr).
\end{align}

We consider the case where $\gamma_k(n) \geq 0$ for all $n$ and $k$. Then, the above updates can be written as  
\vspace{-2ex}
\begin{align}
\zeta(n+1)&=\zeta(n)+a(n)\bigl(h^{(1)}(\zeta,z,\bar\gamma)+M^{(1)}(n+1)\bigr)\\
z(n+1)&=z(n)+b(n)\bigl(h^{(2)}(\zeta,z,\bar\gamma)+M^{(2)}(n+1)\bigr)\\
\gamma_k(n+1)&=\gamma_k(n)+c(n)h^{(3)}(\zeta,z,\bar\gamma).
\end{align}

We begin by verifying \textbf{(A1)-(A3)}. From Assumption \ref{ACC:A3}, \textbf{(A2)} follows. Now consider \textbf{(A1)}. It is easy to see that $h^{(1)}(\zeta,z,\bar\gamma)$ is Lipschitz in $\zeta$ and $\bar\gamma$ since it is linear in $\zeta$ and $\bar\gamma$. Moreover, following the arguments in the previous section, it follows that $h^{(1)}(\zeta,z,\bar\gamma)$ is Lipschitz in $z$ as well. Similarly, $h^{(2)}(\zeta,z,\bar\gamma)$ and $h^{(3)}(\zeta,z,\bar\gamma)$ are Lipschitz in $\zeta$ and $\bar\gamma$ since they are also linear in $\zeta$ and $\bar\gamma$. Again, following the arguments in the previous section, it follows that they are Lipschitz in $z$ as well.

For \textbf{(A3)}, note that from Assumption \ref{ACC_A2}, $\{M^{(1)}(n+1)\}_{n\geq 0}$ and $\{M^{(2)}(n+1)\}_{n\geq 0}$ form Martingale Difference Sequences. Moreover, since the costs $r(s,a)$ and $g_k(s,a)$ along with the feature vectors $\phi_s$ and $\psi_{s,a}$ are bounded for all $(s,a)\in \mathcal{S}\times\mathcal{A}$ (see Assumption \ref{ACC_A1}), it follows that there exists $K_1,K_2>0$ such that $\|M^{(1)}(n+1)\|^2 \leq K_1 (\|\zeta\|^2+\|\bar\gamma\|^2)$ and $\|M^{(2)}(n+1)\|^2 \leq K_2 (\|z\|^2+\|\zeta\|^2+\|\bar\gamma\|^2)$.


Note that for a fixed $z$, define $a^* = \arg \max_a \psi_{s,a}^Tz$  for each state $s$, and let $\pi_z^{\infty}(a|s) = 1_{a=a^*}$. It is known that $\pi_{cz}(a|s) \rightarrow \pi_z^{\infty}(a|s)$ uniformly on compacts \citet{chandru-SB}. Now
\vspace{-2ex}
\begin{align}
    h_{c,\bar\gamma}^{(1)}(\zeta,z)&=\frac{h^{(1)}(c\zeta,cz,\bar\gamma)}{c}=\frac{A^{cz} (c\zeta)+B^{cz}\bar{\gamma}+C^{cz}}{c}=A^{cz}\zeta+\frac{C^{cz}+B^{cz}\bar{\gamma}}{c}.
\end{align}
Let $A^{\infty z} = \sum_s d^{\pi_{z}^{\infty}}(s)\sum_a \pi_{z}^{\infty}(a|s) A_n $. Note that $\lim_{c \rightarrow \infty} (C^{cz}+B^{cz}\bar{\gamma})/c = 0$, $A^{cz} \rightarrow A^{\infty z} $ and $B^{cz}\rightarrow B^{\infty z}$ uniformly on compacts. It follows that $ h_c^{(1)}(\zeta,z,\bar\gamma) \to h_{\infty}^{(1)}(\zeta,z,\bar\gamma):= A^{\infty z}\zeta$ uniformly on compacts as $c \to \infty$. Since $A^{\infty z}$ is negative-definite, the ODE $\dot{\zeta}(t) = h^{(1)}_{\infty}(\zeta(t),z,\bar\gamma)$ has $\lambda^{(1)}_{\infty,\bar\gamma}(z)=0$ as its unique globally asymptotic stable equilibrium point, which is also trivially a Lipschitz continuous function.
\vspace{-2ex}
\begin{align}
    h_{c,\bar\gamma}^{(2)}(z)&=\frac{h^{(2)}(c\lambda_{\infty,\bar\gamma}^{(1)}(z),cz,\bar\gamma)}{c}=\frac{-\epsilon cz + X^{cz} (c\lambda_{\infty,\bar\gamma}^{(1)}(z))+Y^{cz}\bar{\gamma}+Z^{cz}}{c}\\
    &=-\epsilon z +\frac{Y^{cz}\bar{\gamma}+Z^{cz}}{c}.
\end{align}
Note that $\lim_{c \rightarrow \infty} (Y^{cz}\bar{\gamma}+Z^{cz})/c = 0$. It follows that $ h_{c,\bar\gamma}^{(2)}(z) \to h_{\infty,\bar\gamma}^{(2)}(z):= -\epsilon z $ uniformly on compacts as $c \to \infty$. Since $\epsilon>0$, the ODE $\dot{z}(t) = h_{\infty,\bar\gamma}^{(2)}(z(t))$ has $0$ as its unique globally asymptotic stable equilibrium point. Consider the ODE
\vspace{-2ex}
\begin{align}
    \dot\zeta(t)= h^{(1)}(\zeta(t),z,\bar\gamma)=A^{z} \zeta(t)+B^{z}\bar{\gamma}+C^{z}.
\end{align}
It is negative definite with a unique global asymptotic point as $\lambda^{(1)}(z,\bar\gamma)=-(A^{z})^{-1} (B^{z}\bar{\gamma}+C^{z})$, which is the TD limit point of the relaxed problem with fixed policy $\pi_z$ and Lagrange multiplier $\bar\gamma$. Now consider
\begin{align*}
    \dot z(t)&= h^{(2)}(\lambda^{(1)}(z(t),\bar\gamma),z(t),\bar\gamma) = -\epsilon z(t) +( X^{z(t)}(A^{z(t)})^{-1}B^{z(t)}+Y^{z(t)})\bar{\gamma}+Z^{z(t)}\\
    &= -\epsilon z(t)+\nabla_z L(z(t),\bar\gamma).
\end{align*}
The above ODE converges to an $\epsilon$-neighbourhood of the set $\{z:\nabla_z L(z,\bar\gamma)=0\}$. Finally, the ODE
\begin{align}
    \dot\gamma_k(t)= h^{(3)}(\lambda^{(1)}(z(t),\gamma_k),\lambda^{(2)}(\gamma_k),\gamma_k)\overset{(a)}{=}\nabla_\gamma L(z,\bar\gamma).
\end{align}  
where $(a)$ follows from noting that 
$G_{k}(z)-\alpha_{k}=0$ is the same as 
$\nabla_{\gamma}L(z,\bar{\gamma})=0$.  
As with~[12], one can invoke the envelope theorem of mathematical economics
(see pp.\;964--966 of~[21]) to conclude that (32) corresponds to $\dot{\bar\gamma}(t)=\nabla_{\gamma}
         L(z^{\bar{\gamma}(t)},\bar{\gamma}(t))$,
see \cite{bhatnagar2012online} for a more elaborate discussion.  
The above dynamics have $\bar{\gamma}^{*}=\arg\max \bigl(\min_{z}L(z,\cdot)\bigr)$  
as an asymptotically stable equilibrium.  
Thus, $\bar{\gamma}(n)$ converges almost surely to a maximum of the function 
$\min_{z}L(z,\cdot)$.}

\section{Conclusions}
\label{conclusion}

In this work we have provided an easily verifiable set of sufficient conditions for stability and convergence of general $N$-timescale stochastic recursions with a martingale difference noise sequence, along with characterizing the limit of all the $N$ recursions. We then used these results to show that stochastic approximation algorithms with an added heavy ball term in the context of Gradient TD methods can be shown to converge a.s. to the same TD solution asymptotically. There are several directions for further research. A natural direction to pursue, as mentioned in Remark~\ref{rem1}, would be to come up with sufficient conditions for the stability and convergence of $N$-timescale stochastic recursions with Markov noise. 

In \cite{Arunselvan1}, stability conditions for single timescale SA recursions with Markov noise have been provided. The results in \cite{Arunselvan1} also apply to the case where the underlying Markov process does not possess a unique stationary distribution and can in fact depend on not just the underlying parameters but also an additional control sequence. In \cite{prasenjit-SB}, the convergence of two-timescale stochastic approximation with Markov noise is studied assuming that the iterates remain stable. Separately, in \cite{Arunselvan2}, the convergence of two-timescale recursions with set-valued maps is analysed (in the absence of Markov noise) but assuming stability. In \cite{Yaji-SB}, the convergence of more general two-timescale stochastic approximation with set-valued maps and general Markov noise is analysed again assuming iterate-stability. Analysis of algorithms with set-valued maps is important as it paves the way for analysis of RL algorithms under partial observations/information.
Thus, a natural extension of our results will be towards deriving sufficient conditions for both stability and convergence of stochastic approximation with (a) general Markov noise and (b) set-valued maps instead of the usual point-to-point maps as are normally considered and analysed.

 Finally, Section \ref{sec_app} analyzed the asymptotic behaviour of the momentum algorithms. Further analysis of their finite time behaviour is called for to quantify the benefits of using momentum schemes in stochastic approximation. Towards this extension of weak convergence rate analysis of \cite{konda,mokkadem} in the 2-TS setting and recent convergence rate results in expectation and high probability of 2-TS methods in \cite{2TSgugan,harsh,Kaledin,tale2TS} to the $N$-timescale case would also be interesting directions to explore further.

\bibliography{references}
\bibliographystyle{plainnat}
\end{document}